\documentclass[smallextended]{svjour3}       
\smartqed  
\usepackage{graphicx}

\usepackage{latexsym}
\usepackage{amssymb,amsbsy,amsmath,amsfonts,amssymb,amscd,amsfonts,mathrsfs}
\usepackage{graphicx}
\usepackage{color}

\newcommand {\sgn} { {\rm sgn} }

\newcommand{\argmax}[1]{\underset{#1}{\operatorname{arg}\,\operatorname{max}}\;}
\newcommand{\argmin}[1]{\underset{#1}{\operatorname{arg}\,\operatorname{min}}\;}

\newcommand{\beq}{\begin{equation}}
\newcommand{\eeq}{\end{equation}}
\newcommand{\bea} {\begin{array}{rl}}
\newcommand{\eea} {\end{array}}
\newcommand{\bepa}{\left\{ \begin{array}{l}}
\newcommand{\eepa} {\end{array}\right.}

\usepackage{booktabs}

\usepackage{tabularx}

\newtheorem{ass}[theorem]{Assumption}

\begin{document}

\title{Evolutionary dynamics of competing phenotype-structured populations in periodically fluctuating environments
}

\titlerunning{Evolutionary dynamics in fluctuating environments}        

\author{Aleksandra Arda\v{s}eva  \and Robert A Gatenby \and Alexander R A Anderson \and Helen M Byrne \and Philip K Maini \and Tommaso Lorenzi 
}


\institute{
	   Aleksandra Arda\v{s}eva \at
	   Wolfson Centre for Mathematical Biology, Mathematical Institute, 
	   University of Oxford, Andrew Wiles Building, 
	   Radcliffe Observatory Quarter, Woodstock Road, Oxford, OX2 6GG, UK\\
            \email{aleksandra.ardaseva@maths.ox.ac.uk}
	     \and
	     Robert A Gatenby\at
	     Department of Integrated Mathematical Oncology, H. Lee Moffitt Cancer Center, Tampa, Florida, USA\\
	     \email{robert.gatenby@moffitt.org}
	     \and
	     Alexander R A Anderson\at
	     Department of Integrated Mathematical Oncology, H. Lee Moffitt Cancer Center, Tampa, Florida, USA\\
	     \email{alexander.anderson@moffitt.org}
	     \and
             Helen M Byrne \at
	   Wolfson Centre for Mathematical Biology, Mathematical Institute, 
	   University of Oxford, Andrew Wiles Building, 
	   Radcliffe Observatory Quarter, Woodstock Road, Oxford, OX2 6GG, UK\\
             \email{helen.byrne@maths.ox.ac.uk}  
             	     \and
             Philip K Maini \at
	   Wolfson Centre for Mathematical Biology, Mathematical Institute, 
	   University of Oxford, Andrew Wiles Building, 
	   Radcliffe Observatory Quarter, Woodstock Road, Oxford, OX2 6GG, UK\\
             \email{philip.maini@maths.ox.ac.uk}   
             \and	
	    Tommaso Lorenzi \at
	    School of Mathematics and Statistics, University of St Andrews, St Andrews, KY16 9SS, UK\\
            \email{tl47@st-andrews.ac.uk}         
}
\date{Received: date / Accepted: date}

\maketitle

\begin{abstract}
Living species, ranging from bacteria to animals, exist in environmental conditions that exhibit spatial and temporal heterogeneity which requires them to adapt. Risk-spreading through spontaneous phenotypic variations is a known concept in ecology, which is used to explain how species may survive when faced with the evolutionary risks associated with temporally varying environments. In order to support a deeper understanding of the adaptive role of spontaneous phenotypic variations in fluctuating environments, we consider a system of non-local partial differential equations modelling the evolutionary dynamics of two competing phenotype-structured populations in the presence of periodically oscillating nutrient levels. The two populations undergo heritable, spontaneous phenotypic variations at different rates. The phenotypic state of each individual is represented by a continuous variable, and the phenotypic landscape of the populations evolves in time due to variations in the nutrient level. Exploiting the analytical tractability of our model, we study the long-time behaviour of the solutions to obtain a detailed mathematical depiction of the evolutionary dynamics. The results suggest that when nutrient levels undergo small and slow oscillations, it is evolutionarily more convenient to rarely undergo spontaneous phenotypic variations. Conversely, under relatively large and fast periodic oscillations in the nutrient levels, which bring about alternating cycles of starvation and nutrient abundance, higher rates of spontaneous phenotypic variations confer a competitive advantage. We discuss the implications of our results in the context of cancer metabolism.

\keywords{Periodically fluctuating environments \and Evolutionary dynamics \and Spontaneous phenotypic variation \and Bet-hedging \and Non-local partial differential equations}
\end{abstract}

\section{Introduction}
\label{Introduction}
Organisms of various scales, ranging from bacteria to animals, exist in fluctuating environments. For example, in order to cope with changes in nutrient availability, they are required to adapt. When the fluctuations are regular and the populations have sufficient time to sense the changes and react, a highly plastic phenotype, which allows individuals in the population to acquire different traits based on environmental cues, is an optimal strategy~\cite{xue2018benefits}. An alternative strategy that is more suitable for dealing with irregular and unpredictable changes in the environment is risk spreading, which is also known as bet-hedging \cite{cohen1966optimizing}. Here, the population diversifies its phenotypes such that each sub-population is adapted to a specific environment. This ensures that at least some fraction of the population will survive in the face of sudden environmental changes \cite{philippi1989hedging}. Phenotypic heterogeneity, a characteristic feature of a risk spreading strategy, is observed in many systems, including bacterial populations \cite{kussell2005phenotypic} and solid tumours \cite{marusyk2012intra}. 

Bet-hedging is typically proposed to occur within the context of bacterial populations, where experimental support for stochastic phenotype switching is available \cite{kussell2005phenotypic,smits2006phenotypic,veening2008bistability,beaumont2009experimental,acar2008stochastic}. The classic example of bet-hedging is bacterial persistence. During antibiotic treatment a small fraction of slowly growing bacteria, that are resistant to the antibiotic, is able to survive. After the treatment is over, the original population is restored, resulting in resistance to the antibiotic \cite{balaban2004bacterial}. Schreiber et al. \cite{schreiber2016phenotypic} showed that fluctuations in nutrient levels alter the metabolism of bacteria and promote phenotypic heterogeneity. Risk spreading strategies have been observed in other organisms, such as fungi and slime moulds \cite{kussell2005phenotypic}. It is also hypothesized to be present in cancer where irregular vasculature can cause significant fluctuations within the tumour microenvironment \cite{gillies2018eco}. Experimental and theoretical work suggest that intermittent lack of oxygen, \emph{i.e.} cycling hypoxia, leads to clonal diversity, promotes metastasis and selects for more aggressive phenotypes \cite{cairns2001acute,robertson2015impact,chen2018intermittent,nichol2016stochasticity}. 

Mathematically, competition between populations evolving in fluctuating environments has been studied using different modelling approaches, including deterministic predator-prey models and stochastic models \cite{cushing1980two,chesson1994multispecies,roxburgh2004intermediate,anderies2000fluctuating}. Previous work suggests that the likelihood of species coexistence is increased by temporal variations in the environment. More recent models have looked at adaptive strategies, including stochastic phenotype switching, emerging in stochastic environments \cite{fudenberg2012phenotype,muller2013bet,ashcroft2014fixation,wienand2017evolution,gravenmier2018adaptation}. Muller et al. \cite{muller2013bet} theoretically investigated the environmental conditions that would lead to the emergence of bet-hedging, noting the importance of the fluctuation timescales on the success of the adaptation strategy. A rapidly fluctuating environment selects the phenotype that is adapted to averaged conditions, whereas in a slowly varying environment, having two distinct specialists is beneficial. The bet-hedging population was shown to be most successful in an environment that fluctuates on an intermediate timescale. 

Most of the experimental and theoretical models that have been developed to explore the dynamics of phenotypic changes in fluctuating environments consider the state of the environment and the phenotypic state of the individuals to be binary -- \emph{i.e.} the environment switches between two extreme conditions and individuals are allowed to jump between two antithetical phenotypic states that are each adapted to opposing environmental conditions~\cite{acar2008stochastic,muller2013bet,wienand2017evolution}. However, in many cases of biological and ecological interest {\it Natura non facit saltus}, and it might therefore be relevant to consider the occurrence of intermediate environmental conditions and the existence of a spectrum of possible phenotypic states. 

In light of these considerations, we present here a novel mathematical model for the evolutionary dynamics of two competing phenotype-structured populations in periodically fluctuating environments. The phenotypic state of each individual is represented by a continuous variable, and the phenotypic fitness landscape of the populations evolves in time due to variations in the concentration of a nutrient. In order to assess the evolutionary role that heritable, spontaneous phenotypic variations play in environmental adaptation, we focus on the case where the two populations undergo such phenotypic variations with different probabilities.

In our model, the phenotype distribution of the individuals within each population is described by a population density function that is governed by a parabolic partial differential equation (PDE), whereby a linear diffusion operator models the occurrence of spontaneous phenotypic variations, while a non-local reaction term takes into account the effects of asexual reproduction and intrapopulation competition. The two non-local parabolic PDEs for the population density functions are coupled through an additional non-local term modelling the effects of interpopulation competition. In such a mathematical framework, the fact that the two populations undergo phenotypic variations with different probabilities translates into the assumption that the two PDEs have different diffusion coefficients. 

Mathematical models formulated in terms of integrodifferential equations and non-local parabolic PDEs like those considered here have been increasingly used to achieve a more in-depth theoretical understanding of the mechanisms underlying phenotypic adaptation in a variety of biological contexts~\cite{alfaro2018evolutionary,alfaro2017effect,almeida2019evolution,bouin2014travelling,bouin2012invasion,busse2016mass,calsina2004small,calsina2007asymptotic,calsina2013asymptotics,chisholm2016evolutionary,chisholm2015emergence,cuadrado2009equilibria,delitala2012mathematical,delitala2012asymptotic,delitala2013mathematical,desvillettes2008selection,diekmann2005dynamics,domschke2017structured,iglesias2018long,lam2017dirac,lam2017stability,lorenzi2014asymptotic,lorenzi2016tracking,lorenzi2015dissecting,lorenzi2015mathematical,lorenzi2018role,lorz2015modeling,lorz2013populational,mirrahimi2015time,Nordmann2018,perthame2006transport,pouchol2018asymptotic,turanova2015model}. In particular, our work follows earlier papers on non-local parabolic PDEs modelling the evolutionary dynamics of populations structured by continuous traits in periodically-fluctuating environments~\cite{iglesias2018long,lorenzi2015dissecting}. Compared to these previous studies, which considered scalar equations modelling the dynamics of single population, our model comprises a system of coupled equations modelling the dynamics of competing populations. This requires a novel extension of the methods developed in~\cite{lorenzi2015dissecting} to characterise the qualitative and quantitative properties of the solutions.

Exploiting the analytical tractability of our model, we study the long-time behaviour of the solutions in order to obtain a detailed mathematical depiction of the evolutionary dynamics. Moreover, the asymptotic results are compared to numerical solutions of the model equations. Our analytical and numerical results clarify the role of heritable, spontaneous phenotypic variations as drivers of adaptation in periodically fluctuating environments. 

The paper is organised as follows. In Section~\ref{Model} we introduce the mathematical model. In Section~\ref{Section3} we carry out an analytical study of evolutionary dynamics. In Section~\ref{Section4} we integrate the analytical results with numerical simulations. In Section~\ref{Section5} we discuss the biological relevance of our theoretical findings in the context of cancer cell metabolism and tumour-microenvironment interactions. Section~\ref{Conclusions} concludes the paper and provides a brief overview of possible research perspectives.

\section{Model description}
\label{Model}
We study the evolutionary dynamics of two competing phenotype-structured populations in a well-mixed system. Individuals within the two populations reproduce asexually, die and undergo heritable, spontaneous phenotypic variations. We assume the two populations  differ only in the rate at which they undergo phenotypic variations. We label the population undergoing phenotypic variations at a higher rate by the letter $H$, while the other population is labelled by the letter $L$.

We represent the phenotypic state of each individual by a continuous variable $x \in \mathbb{R}$, and we describe the phenotype distributions of the two populations at time $t \in [0,\infty)$ by means of the population density functions $n_H(x,t)~\geq~0$ and $n_L(x,t)~\geq~0$. We define the size of population $H$, the size of population $L$ and the total number of individuals inside the system at time $t$, respectively, as
\beq
\label{erhos}
\rho_{H}(t) = \int_{\mathbb{R}} n_{H}(x,t) \; {\rm d}x, \quad \rho_{L}(t) = \int_{\mathbb{R}} n_{L}(x,t) \; {\rm d}x ,\quad \rho(t) = \rho_H(t) + \rho_L(t).
\eeq
Moreover, we define, respectively, the mean phenotypic state and the related variance of each population $i \in \{H,L\}$ at time $t$ as
\beq
\label{emusigma}
\mu_{i}(t) = \frac{1}{\rho_i(t)} \int_{\mathbb{R}} x \, n_{i}(x,t) \; {\rm d}x, \quad \sigma^2_{i}(t)  = \frac{1}{\rho_i(t)} \int_{\mathbb{R}} x^2 \; n_i(x,t) \; {\rm d}x -  \mu_i^2(t).
\eeq
In the mathematical framework of our model, the function $\sigma^2_{i}(t)$ provides a measure of the level of phenotypic heterogeneity in the $i^{th}$ population. Finally, we introduce a function $S(t) \geq 0$ to model the concentration of a nutrient that is equally available to the two populations at time $t$, which we assume is given.

The evolution of the population density functions is governed by the following system of non-local parabolic PDEs
\begin{equation}
\label{eS1}
\left\{
\begin{array}{ll}
\displaystyle{\frac{\partial n_H}{\partial t} = \beta_H \frac{\partial^2 n_H}{\partial x^2} \; + \; R\big(x,S(t),\rho(t)\big) \, n_H}, 
\\\\
\displaystyle{\frac{\partial n_L}{\partial t} = \beta_L \frac{\partial^2 n_L}{\partial x^2} \;\;\; + \; R\big(x,S(t),\rho(t)\big) \, n_L},
\end{array}
\right.
\quad  \text{for } (x,t) \in \mathbb{R} \times (0,\infty).
\end{equation} 

In the system of PDEs \eqref{eS1}, the diffusion terms model the effects of spontaneous phenotypic variations, which occur at rates $\beta_H >0$ and $\beta_L > 0$, with 
\begin{equation}
\label{assbeta}
\beta_H > \beta_L.
\end{equation} 
The functional $R\big(x,S(t),\rho(t)\big)$ models the fitness of individuals in the phenotypic state $x$ at time $t$ under the environmental conditions given by the nutrient concentration $S(t)$ and the total number of individuals $\rho(t)$ -- \emph{i.e.} the functional $R\big(x,S(t),\rho(t)\big)$ can be seen as the phenotypic fitness landscape of the two populations at time $t$. Throughout the paper, we define this fitness functional as
\begin{equation}\label{R}
R\big(x,S,\rho\big) = p(x,S) \; - \; d  \rho.
\end{equation}
Definition~\eqref{R} translates into mathematical terms the following biological ideas: (i) all else being equal, individuals die due to nterpopulation and intrapopulation competition at rate $d \rho(t)$, with the parameter $d > 0$ being related to the carrying capacity of the system in which the two populations are contained; (ii) individuals in the phenotypic state $x$ proliferate and die under natural selection at rate $p(x,S(t))$ (\emph{i.e.} the function $p(x,S)$ is a net proliferation rate). We focus on a scenario corresponding to the biological assumptions given hereafter.

\begin{ass}
\label{bioass1}
Phenotypic variants with $x \to 0$ have a competitive advantage over the other phenotypic variants when the nutrient concentration is high (i.e. if $S(t) \gg 1$). 
\end{ass}
\begin{ass}
\label{bioass2}
Phenotypic variants with $x \to 1$ are favoured over the other phenotypic variants when the nutrient concentration is low (i.e. if $S(t) \ll 1$). 
\end{ass}
Under Assumptions \ref{bioass1} and \ref{bioass2}, we define the net proliferation rate as
\begin{equation}\label{pn}
p\big(x,S(t)\big) = \gamma \, \frac{S(t)}{1 + S(t)} \, \left(1 - x^2 \right) \; + \; \zeta \, \left(1 - \frac{S(t)}{1 + S(t)}\right) \left[1 - \left(1-x\right)^2 \right],
\end{equation}
with $0 < \zeta \leq \gamma$. The parameters $\gamma$ and $\zeta$ model, respectively, the maximum proliferation rate of the phenotypic variants best adapted to nutrient-rich and nutrient-scarce environments.

Definition~\eqref{pn} ensures analytical tractability of the model and leads to a fitness functional that is close to the approximate fitness landscapes which can be inferred from experimental data through regression techniques -- see, for instance, equation (1) in~\cite{otwinowski2014inferring}. In fact, after a little algebra, definition~\eqref{pn} can be rewritten as
\begin{equation}\label{pn2ori}
p\big(x,S\big) = \gamma \, g(S) - h(S) \left(x - \varphi(S) \right)^2
\end{equation}
with 
\begin{equation}\label{ghori}
g(S) = \frac{1}{1 + S} \left(S + \frac{\zeta}{\gamma} \frac{\zeta}{\zeta + \gamma S}\right), \quad \varphi(S) = \frac{\zeta}{\zeta + \gamma \, S}
\end{equation}
and
\begin{equation}\label{hori}
h(S)= \zeta + (\gamma - \zeta) \frac{S}{1+S}.
\end{equation}
Under the environmental conditions defined by the nutrient concentration $S$, the function $0 \leq \varphi(S) \leq 1$ represents the fittest phenotypic state, $\gamma \, g(S) > 0$ is the maximum fitness, and $h(S)$ can be seen as a nonlinear selection gradient that quantifies the intensity of natural selection. Throughout the paper we will refer to $g(S)$ as the rescaled maximum fitness. 

In accordance with Assumptions \ref{bioass1} and \ref{bioass2}, equation~\eqref{pn2ori} shows that definition~\eqref{pn} is such that the fittest phenotypic state $\varphi(S)$ belongs to the interval $[0,1]$ for any nutrient concentration $S\geq 0$, \emph{i.e.} $\varphi : \mathbb{R}_{\geq 0} \to [0,1]$. In particular, under starvation conditions (\emph{i.e.} if $S = 0$) the fittest phenotypic state is $\varphi(0)=1$, while increasing nutrient concentrations correspond to values of the fittest phenotypic state closer to $0$, {\it i.e.} $\varphi'(S)<0$ for all $S \geq 0$ and $\varphi(S) \to 0$ as $S \to \infty$. Furthermore, the fact that the function $p(x,S)$ is negative for values of $x$ sufficiently far from the fittest phenotypic state $\varphi(S)$ captures the idea that less fit variants are driven to extinction by natural selection. These observations are illustrated by the plots in Figure~\ref{fig1}. 
\begin{figure}[h] \centering
\includegraphics[width=0.9 \linewidth]{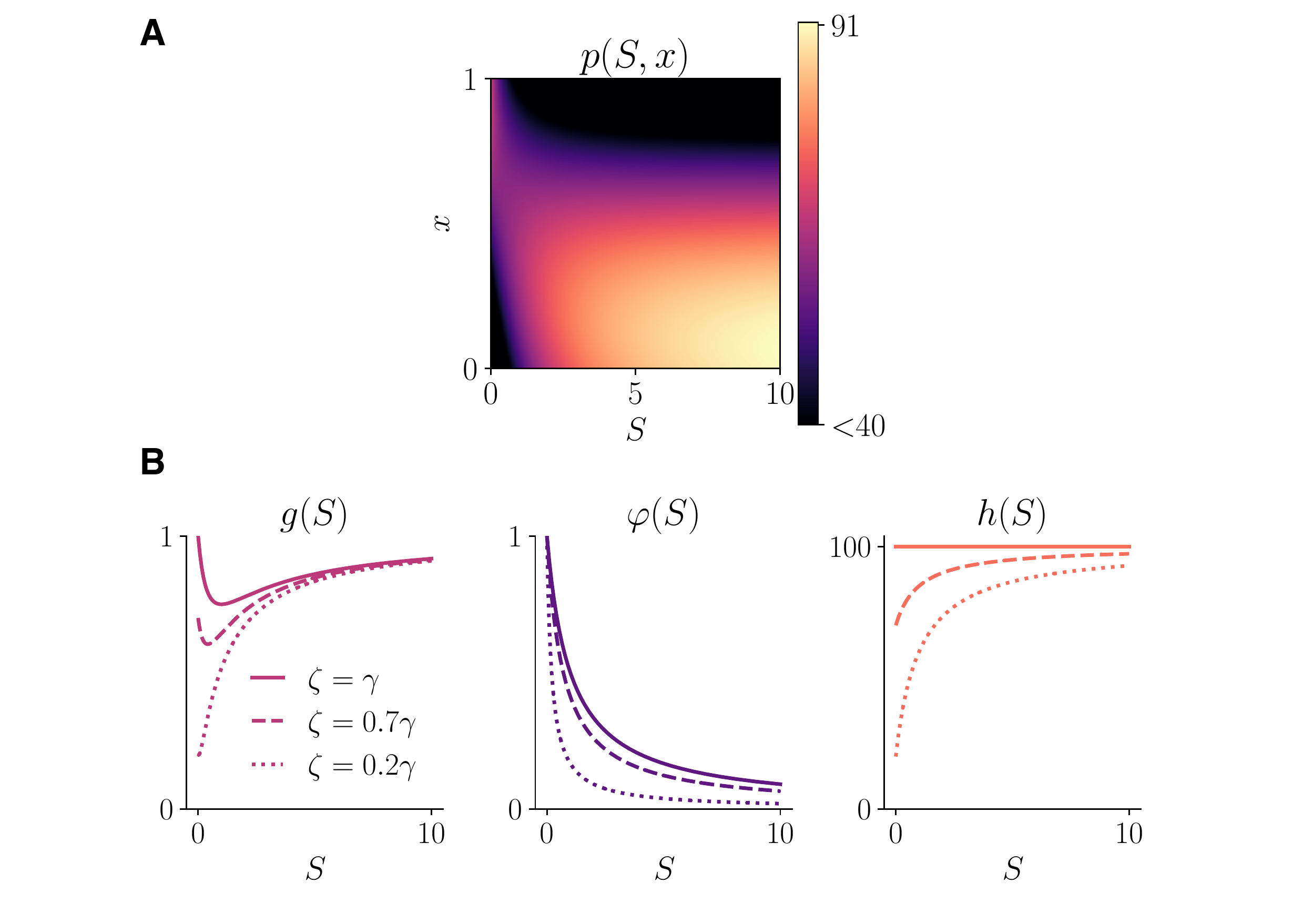}
\caption{\label{fig1} {\bf {\textsf A}.} Plot of the net proliferation rate $p(S,x)$ defined by~\eqref{pn} [or equivalently by \eqref{pn2ori}] with $\gamma = 100$ and $\zeta = 50$. {\bf {\textsf B}.} The rescaled maximum fitness $g(S)$ and the fittest phenotypic state $\varphi(S)$ defined by~\eqref{ghori}, along with the selection gradient $h(S)$ defined by~\eqref{hori}, are plotted against the nutrient concentration $S$, for $\gamma = 100$ and different values of the parameter $\zeta$. In this paper we consider the case $\zeta=\gamma$}
\end{figure}

Henceforth for simplicity we assume
\beq
\label{asimp}
\zeta = \gamma.
\eeq
Under assumption~\eqref{asimp}, definitions~\eqref{ghori} and~\eqref{hori} become, respectively,
\beq
\label{gh}
g(S) = \frac{1}{1 + S} \left(S + \frac{1}{1 + S}\right), \quad \varphi(S) = \frac{1}{1 + S} \quad \text{and} \quad h(S) \equiv \gamma.
\eeq
Moreover, since we assume the function $S(t)$ to be given, we use the notation
$$
g(t)\equiv g(S(t)) \quad \text{and} \quad  \varphi(t)\equiv \varphi(S(t)).
$$

\section{Analysis of evolutionary dynamics}
\label{Section3}
To obtain an analytical description of the evolutionary dynamics, we focus on a biological scenario whereby the initial phenotype distributions of the two populations are Gaussians, that is, we study the behaviour of the solution to the system of non-local parabolic equations~\eqref{eS1} subject to the initial condition given by the pair $n_H(x,0)$ and $n_L(x,0)$ with
\begin{equation}
\label{eS1S0ic}
n_i(x,0) = \rho_i^0 \sqrt{\frac{v^0_i}{2 \pi}} \, \exp\left[-\frac{v^0_i}{2} \left(x - \mu^0_i \right)^2\right] \quad \text{for} \quad  i \in \{H,L\},
\end{equation} 
where $\rho^0_i \in \mathbb{R}_{>0}$, $v^0_i \in \mathbb{R}_{>0}$ and $\mu^0_i \in \mathbb{R}$. 
\begin{remark}
The choice of initial condition~\eqref{eS1S0ic} is consistent with much of the previous work on the mathematical analysis of the evolutionary dynamics of continuous traits, which relies on the prima facie assumption that population densities are Gaussians~\cite{rice2004evolutionary}. 
\end{remark}

Before turning to the case of periodically fluctuating environments in Section~\ref{Section32}, we consider the case of constant environments in Section~\ref{Section31}. The proofs of the results presented in Section~\ref{Section31} and Section~\ref{Section32} rely on the results established by the Proposition~\ref{Prop1}.

\begin{proposition}
\label{Prop1}
Under assumptions~\eqref{R}, \eqref{pn2ori} and \eqref{gh}, the system of non-local PDEs~\eqref{eS1} subject to the initial condition~\eqref{eS1S0ic} admits the exact solution
\begin{equation} \label{gaussian}
n_i(x,t)= \rho_i(t) \sqrt{\frac{v_i(t)}{2 \pi}} \exp\left[-\frac{v_i(t)}{2} (x-\mu_i(t))^2  \right] \quad \text{for} \quad i \in \left\{H, L \right\},
\end{equation}
with the population size, $\rho_i(t)$, the mean phenotypic state, $\mu_i(t)$, and the inverse of the related variance, $v_i(t) = 1/\sigma^{2}_i(t)$, being solutions of the Cauchy problem
\begin{equation}
\label{odes}
\left\{
\begin{array}{ll}
\displaystyle{v_i'(t) = 2 \left(\gamma - \beta_i v_i^2(t)\right)},
\\\\
\displaystyle{\mu_i'(t) = \frac{2\gamma}{v_i(t)} \left(\varphi(t) - \mu_i(t)\right)}, 
\\\\
\displaystyle{\rho_i'(t) = \left(F_i(t) - d \rho(t) \right)\rho_i(t)},
\\\\
v_i(0) = v^0_i, \quad \mu_i(0) = \mu^0_i, \quad \rho_i(0) = \rho^0_i, 
\\\\
\rho(t) = \rho_H(t) + \rho_L(t),
\end{array}
\right.
\qquad 
\text{for } \; i \in \left\{H, L \right\},
\end{equation} 
where
\begin{equation} \label{Fit}
F_i(t) \equiv F_i(t,v_i(t),\mu_i(t)) =  \gamma \, g(t) - \frac{\gamma}{v_i(t)} - \gamma \left( \mu_i(t) - \varphi(t) \right)^2.
\end{equation}
\end{proposition}

\begin{proof}
Substituting definitions~\eqref{R}, \eqref{pn2ori} and \eqref{gh} into the non-local PDE \eqref{eS1} for $n_i(x,t)$ yields
\begin{equation}
\label{e1-multi}
\displaystyle{\frac{\partial n_i}{\partial t} = \beta_i \frac{\partial^2 n_i}{\partial x^2} \; + \; \left[\gamma g(t) - \gamma (x - \varphi(t))^2 - d \rho(t) \right] n_i}, \quad n_i \equiv n_i(x,t).
\end{equation} 
Building upon the results presented in~\cite{almeida2019evolution,chisholm2016evolutionary,lorenzi2015dissecting}, we make the ansatz~\eqref{gaussian} and substituting this ansatz into equation~\eqref{e1-multi} we find
\begin{eqnarray}\label{eAnsatzTest}
\frac{\rho_i'}{\rho_i}+\frac{v_i'}{2v_i} &=& \frac{v_i'}{2}\left(x-\mu_i\right)^2 - \mu_i' \, v_i \left(x-\mu_i\right) + \,  \beta_i\left[v_i^2 \left(x-\mu_i\right)^2 -v_i \right] \nonumber \\
&& + \,  \gamma \, g(t)- \gamma \, \left(x-\varphi(t)\right)^2 - d \rho.
\end{eqnarray} 
Equating the coefficients of the zero-order, first-order and second-order terms in $x$ in \eqref{eAnsatzTest} produces a system of differential equations. Namely, the second-order terms in $x$ yield the following differential equation for $v_i$ alone
\beq
\label{odevi}
v_i' + 2 \beta_i v_i^2= 2 \gamma.
\eeq
Moreover, equating the coefficients of the first-order terms in $x$, and eliminating $v'_i$ from the resulting equation, yields
\beq
\label{odemui}
\mu_i' = \frac{2 \gamma (\varphi - \mu_i)}{v_i}.
\eeq
Lastly, choosing $x=\mu_i$ in equation~\eqref{eAnsatzTest} gives 
\begin{equation}\label{eAnzatzOneDE}
\frac{\rho_i'}{\rho_i}+\frac{v_i'}{2v_i}=-\beta_i v_i + \gamma \, g - \gamma (\mu_i - \varphi)^2 - d \rho
\end{equation}
and eliminating $v'_i$ from equation~\eqref{eAnzatzOneDE} we find
\begin{equation} 
\label{oderhoi}
\rho_i' = \left( F_i - d \rho \right) \rho_i,
\end{equation}
with the function $F_i(t)$ being defined according to~\eqref{Fit}. Under the initial condition~\eqref{eS1S0ic}, we have
$$
v_i(0)=v^0_i, \quad \mu_i(0)=\mu^0_i \quad \text{and} \quad \rho_i(0)=\rho^0_i.
$$
Imposing these initial conditions on the system of differential equations~\eqref{odevi}-\eqref{oderhoi}, we arrive at the Cauchy problem \eqref{odes} for the functions $v_i(t)$, $\mu_i(t)$ and $\rho_i(t)$.
\qed
\end{proof}

\subsection{Evolutionary dynamics in constant environments}
\label{Section31}
Focussing on the case of constant environments, we let the nutrient concentration be constant and thus we make the assumption 
\beq
\label{aS0}
S(t) \equiv \overline{S} \geq 0,
\eeq
which implies that 
\beq
\label{gh0}
g(t) \equiv \overline{g} \quad \text{and} \quad \varphi(t) \equiv \overline{\varphi}.
\eeq
In this case, our main results are summarised by Theorem~\ref{theorem1}.
\begin{theorem} \label{theorem1}
Under assumptions~\eqref{assbeta}, \eqref{R}, \eqref{pn2ori}, \eqref{gh} and the additional assumption \eqref{aS0}, the solution of the system of PDEs \eqref{eS1} subject to the initial condition~\eqref{eS1S0ic} is of the Gaussian form \eqref{gaussian} and satisfies the following:
\begin{itemize}
\item[(i)] if 
\beq
\label{Th1assi}
\sqrt{\beta_L} \geq \sqrt{\gamma} \, \overline{g}
\eeq
then 
\beq
\label{Th1i}
\lim_{t\to\infty}\rho_H(t)=0 \quad \text{and} \quad \lim_{t\to\infty}\rho_L(t)=0;
\eeq
\item[(ii)] if 
\beq
\label{Th1assii}
\sqrt{\beta_L} < \sqrt{\gamma} \, \overline{g}
\eeq
then 
\beq
\label{Th1ii}
\lim_{t\to\infty}\rho_H(t)=0, \quad \lim_{t\to\infty}\rho_L(t)= \frac{\sqrt{\gamma}}{d} \left( \sqrt{\gamma} \, \overline{g} -  \sqrt{\beta_L} \right)
\eeq
and
\begin{equation}\label{Th1iii}
\lim_{t\to\infty}\mu_L(t) = \overline{\varphi}, \quad \lim_{t\to\infty} \sigma^2_L(t) = \sqrt{\frac{\beta_L}{\gamma}}.
\end{equation}
\end{itemize}
\end{theorem}

\begin{proof}
Under the additional assumption \eqref{aS0}, Proposition~\ref{Prop1} ensures that the population density function $n_i(x,t)$ is of the Gaussian form \eqref{gaussian} with the population size, $\rho_i(t)$, the mean phenotypic state, $\mu_i(t)$, and the inverse of the related variance, $v_i(t)=1/\sigma_i^2(t)$, being governed by the Cauchy problem~\eqref{odes} with $g(t) \equiv \overline{g}$ and $\varphi(t) \equiv \overline{\varphi}$.

In this framework, we divide the proof of Theorem~\ref{theorem1} into four steps. We study the asymptotic behaviour of $v_i(t)$, $\mu_i(t)$ and $F_i(t)$ for $t \to \infty$ (Step~1). We show that $\rho_i(t)$ is non-negative and uniformly bounded (Step~2). Finally, we prove claim~\eqref{Th1i} (Step~3), and we conclude with the proof of claims~\eqref{Th1ii} and \eqref{Th1iii} (Step~4).

\paragraph{Step 1: asymptotic behaviour of $v_i(t)$, $\mu_i(t)$ and $F_i(t)$ for $t \to \infty$.} Solving the separable first-order differential equation~\eqref{odes}$_1$ for $v_i(t)$ and imposing the initial condition~\eqref{odes}$_4$ gives
\begin{equation}
\label{vi}
v_i(t) = \sqrt{\frac{\gamma}{\beta_i}} \, \frac{\sqrt{\gamma/\beta_i}+v^0_i - \left(\sqrt{\gamma/\beta_i} - v^0_i\right)\exp\left(-4\sqrt{\gamma \, \beta_i} \, t\right)}
{\sqrt{\gamma/\beta_i}+v^0_i + \left(\sqrt{\gamma/\beta_i} - v^0_i\right)\exp\left(-4\sqrt{\gamma \beta_i} \, t\right)},
\end{equation}  
which implies that
\begin{equation}
\label{viinf}
v_i(t) \to \sqrt{\frac{\gamma}{\beta_i}} \quad \text{exponentially fast} \text{ as } t \to \infty.
\end{equation} 
Moreover, solving the differential equation~\eqref{odes}$_2$ for $\mu_i(t)$ by the integrating factor method and imposing the initial condition~\eqref{odes}$_4$ yields 
\begin{equation}
\label{mui}
\mu_i(t) = \mu^0_i \exp\left(-2\gamma\int_0^t \frac{{\rm d}s}{v_i(s)}\right) + \overline{\varphi} \left[1 - \exp\left(-2\gamma\int_0^t \frac{{\rm d}s}{v_i(s)}\right) \right] ,
\end{equation} 
from which, using the positivity of $v_i(t)$, we find that
\begin{equation}
\label{muiinf}
\mu_i(t)  \to \overline{\varphi} \quad \text{exponentially fast} \text{ as } t \to \infty.
\end{equation} 
Lastly, noting that, under the additional assumption \eqref{aS0}, the function $F_i(t)$ defined by~\eqref{Fit} reads as
 \beq
\label{FitS0}
F_i(t) = \gamma \, \overline{g} - \frac{\gamma}{v_i(t)} - \gamma \left(\mu_i(t) - \overline{\varphi} \right)^2,
\eeq
the asymptotic results~\eqref{viinf} and \eqref{muiinf} allow us to conclude that 
\begin{equation}
\label{Fiinf}
F_i(t)  \to \gamma \, \overline{g} - \sqrt{\gamma \, \beta_i} \quad \text{exponentially fast} \text{ as } t \to \infty.
\end{equation}

\paragraph{Step 2: non-negativity and boundedness of $\rho_i(t)$.} Solving the differential equation~\eqref{odes}$_3$ for $\rho_i$ and imposing the initial condition~\eqref{odes}$_4$ yields
\begin{equation}
\label{rhoi}
\rho_i(t) = \rho^0_i \exp\left[\int_0^t \left(F_i(s) - d \rho(s) \right) {\rm d}s\right].
\end{equation}
This result, along with the positivity of $\rho^0_i$, implies that
\beq
\label{LBi}
\rho_i(t) \geq 0 \quad \text{for all } t \geq 0.
\eeq
Moreover, substituting~\eqref{FitS0} into the differential equation~\eqref{odes}$_3$ for $\rho_i$ yields
$$
\rho'_i(t) = \left[\gamma \, \overline{g} - \frac{\gamma}{v_i(t)} - \gamma \left ( \mu_i(t) - \overline{\varphi} \right)^2 \right] \rho_i(t) - d \left(\rho_i(t) + \rho_j(t)\right) \, \rho_i(t),
$$
with $j=L$ if $i=H$ and $j=H$ if $i=L$. Estimating from above the right-hand side of the latter differential equation by using the non-negativity of $\rho_j(t)$ [\emph{cf.} the uniform lower bound~\eqref{LBi}], the positivity of $v_i(t)$ [\emph{cf.} expression~\eqref{vi}] and the fact that $\overline{g} < 2$ [\emph{cf.} definition~\eqref{gh}], we obtain the differential inequality
$$
\rho'_i(t) \leq \left(2 \gamma - d \, \rho_i(t) \right) \rho_i(t),
$$
which gives the uniform upper bound
\beq
\label{UBiS0}
\rho_i(t) \leq \max \left\{\rho_i^0, \frac{2 \, \gamma}{d} \right\} \quad \text{for all } t \geq 0.
\eeq

\paragraph{Step 3: proof of claim~\eqref{Th1i}.} Combining the asymptotic result \eqref{Fiinf} with the expression~\eqref{rhoi} for $\rho_i$ we find that
\begin{equation}
\label{rhoiinf}
\rho_i(t) \sim C \rho^0_i \exp\left[\left(\gamma \overline{g} - \sqrt{\gamma \, \beta_i}\right) t - d \int_0^t \rho(s) \, {\rm d}s \right] \quad \text{as } \; t \to \infty,
\end{equation}
for some positive constant $C$. Since the function $\rho(t)$ is non-negative [\emph{cf.} the uniform lower bound~\eqref{LBi}], the asymptotic relation~\eqref{rhoiinf} ensures that
\beq
\label{extinction}
\text{if} \quad \sqrt{\beta_i} \geq \sqrt{\gamma} \, \overline{g} \quad \text{then} \quad \lim_{t\to\infty} \rho_i(t) = 0.
\eeq
Under assumption~\eqref{assbeta} and the additional assumption~\eqref{Th1assi}, claim~\eqref{Th1i} follows from the asymptotic result~\eqref{extinction}.

\paragraph{Step 4: proof of claims~\eqref{Th1ii} and \eqref{Th1iii}.} As long as $\rho_H(t)>0$, we can compute the quotient of $\rho_L(t)$ and $\rho_H(t)$ through~\eqref{rhoi}. In so doing we find
\begin{equation}
\label{quotrhoi}
\frac{\rho_L(t)}{\rho_H(t)} = \frac{\rho^0_L}{\rho^0_H} \exp\left[\int_0^t \left(F_L(s) - F_H(s) \right) {\rm d}s\right].
\end{equation}
Using the limit~\eqref{Fiinf} for $F_i$, we then have
\begin{equation}
\label{quotrhoiinf}
\frac{\rho_L(t)}{\rho_H(t)} \sim C \exp\left[\sqrt{\gamma} \, \left(\sqrt{\beta_H} - \sqrt{\beta_L}\right) t \right] \quad \text{as } \; t \to \infty,
\end{equation}
for some positive constant $C$. Under assumption~\eqref{assbeta}, the asymptotic relation~\eqref{quotrhoiinf} gives
$$
\lim_{t \to \infty} \frac{\rho_L(t)}{\rho_H(t)} = \infty
$$
and, since $\rho_L$ is uniformly bounded from above [\emph{cf.} the uniform upper bound~\eqref{UBiS0}], from~\eqref{quotrhoiinf} we conclude that
\begin{equation}
\label{rhoHto0}
\rho_H(t) \to 0 \quad \text{exponentially fast} \text{ as } t \to \infty.
\end{equation}
We can rewrite the differential equation~\eqref{odes}$_3$ for $\rho_L$ as
\begin{equation}
\label{oderhoLs0}
\rho_L'(t) = \left[\left(\gamma \overline{g} - \sqrt{\gamma \, \beta_L} + \eta(t)\right) - d \rho_L(t)\right] \rho_L(t),
\end{equation}
where the function $\eta(t)$ is defined as 
$$
\eta(t) = \left(\sqrt{\gamma \, \beta_L} - \frac{\gamma}{v_L(t)} \right) - \gamma \left(\mu_L(t) - \overline{\varphi} \right)^2 - d \rho_H(t). 
$$
Using the asymptotic results~\eqref{viinf},~\eqref{muiinf} and~\eqref{rhoHto0}, we see that
\begin{equation}
\label{newresetatozero}
\eta(t)  \to 0 \quad \text{exponentially fast} \text{ as } t \to \infty.
\end{equation}
Solving the differential equation~\eqref{oderhoLs0} complemented with the initial condition $\rho_L(0)=\rho_L^0$ yields~\cite{chisholm2016evolutionary} 
\begin{equation}
\label{rhoLs0}
\rho_L(t) = \frac{\displaystyle \rho^0_L \exp\left[\int_0^t \left(\gamma \overline{g} - \sqrt{\gamma \, \beta_L} + \eta(s)\right) {\rm d}s\right]}{\displaystyle 1 + d \, \rho^0_L \int_0^t \exp\left[\int_0^s \left(\gamma \overline{g} - \sqrt{\gamma \, \beta_L} + \eta(z)\right) {\rm d}z\right]{\rm d}s}.
\end{equation}
The result~\eqref{newresetatozero} ensures that in the asymptotic regime $t \to \infty$ we have
$$
\exp\left[\int_0^t \left(\gamma \overline{g} - \sqrt{\gamma \, \beta_L} + \eta(s)\right) {\rm d}s\right] \sim C \, \exp\left[\left(\gamma \overline{g} - \sqrt{\gamma \, \beta_L} \right) t \right]
$$
and, under the additional assumption~\eqref{Th1assii}, we also have 
$$
 \int_0^t \exp\bigg[\int_0^s \bigg(\gamma \overline{g} - \sqrt{\gamma \, \beta_L} + \eta(z)\bigg) {\rm d}z\bigg] {\rm d}s \sim
 C \, \frac{\exp\left[{\left(\gamma \overline{g} - \sqrt{\gamma \, \beta_L}\right) \, t}\right]}{\gamma \overline{g} - \sqrt{\gamma \, \beta_L}},
$$
for some positive constant $C$. These asymptotic relations, along with the expression~\eqref{rhoLs0} for $\rho_L$, allow us to conclude that
\begin{equation}
\label{rhoLtobarrhoL}
\lim_{t\to\infty} \rho_L(t) = \frac{\gamma \overline{g} - \sqrt{\gamma \, \beta_L}}{d}.
\end{equation}
Claims~\eqref{Th1ii} and \eqref{Th1iii} follow from the asymptotic results~\eqref{rhoHto0} and \eqref{rhoLtobarrhoL}, and the asymptotic results~\eqref{muiinf} and~\eqref{viinf} with $i=L$. 
\qed
\end{proof}

The asymptotic results established by Theorem~\ref{theorem1} provide a mathematical formalisation of the idea that in constant environments:
\begin{enumerate}
\item populations undergoing spontaneous phenotypic variation at a rate that is too large compared to the maximum fitness will ultimately go extinct [\emph{cf.} point (i)];
\item {\it ceteris paribus}, if at least one population undergoes spontaneous phenotypic variation at a rate sufficiently small compared to the maximum fitness [\emph{cf.} point (ii)] then: 
\begin{enumerate}
\item[2(a).] the population with the lower rate of phenotypic variation will outcompete the other population;
\item[2(b).] the equilibrium phenotype distribution of the surviving population will be unimodal with the mean phenotype corresponding to the fittest phenotypic state and the related variance being directly proportional to the rate of phenotypic variations.
\end{enumerate}
\end{enumerate}

\subsection{Evolutionary dynamics in periodically fluctuating environments}
\label{Section32}
We now focus on the case of environments that undergo fluctuations with period $T>0$, and we assume the nutrient concentration to be Lipschitz continuous and $T$-periodic, \emph{i.e.} we let $S : [0,\infty) \to \mathbb{R}_{\geq 0}$ satisfy the assumptions 
\beq
\label{aST}
S \in {\rm Lip}([0,\infty)) \quad \text{and} \quad S(t+T) = S(t) \quad \text{for all } t \geq 0,
\eeq 
which implies that the functions $g(t)$ and $\varphi(t)$ satisfy the assumptions
\beq
\label{gphiSTa}
g, \varphi \in {\rm Lip}([0,\infty)), \;\;  g(t+T) = g(T) \;\; \text{and} \;\; \varphi(t+T) = \varphi(T) \;\; \text{for all } t \geq 0.
\eeq
Our main results are summarised by Theorem~\ref{theorem2}, the proof of which relies on the results established by Lemma~\ref{Lemma1} and Lemma~\ref{Lemma2}. 
\begin{lemma}
\label{Lemma1}
Under assumptions~\eqref{R}, \eqref{pn2ori}, \eqref{gh} and \eqref{aST}, the unique real $T$-periodic solution of the problem
\begin{equation}
\label{uieq}
\begin{cases}
u_i'(t) = 2 \, \sqrt{\gamma \, \beta_i} \, \left(\varphi(t) - u_i(t) \right), \quad \text{for } t \in (0,T),
\\
u_i(0) = u_i(T),
\end{cases}
\end{equation}
is
\begin{align}
\label{uiTh2}
u_i(t) &= \frac{2\sqrt{\gamma\beta_i}\exp\left(-2\sqrt{\gamma\beta_i}t\right)}{\exp\left(2 \, \sqrt{\gamma\beta_i} \, T\right)-1}\int_0^T\exp\left(2 \, \sqrt{\gamma\beta_i} \, s\right)\varphi(s) \, {\rm d}s \nonumber \\
& \quad +2\sqrt{\gamma\beta_i}\exp\left(-2 \sqrt{\gamma\beta_i} \, t\right)\int_0^t\exp\left(2 \sqrt{\gamma\beta_i} \, s\right)\varphi(s) \, {\rm d}s,
\end{align}
and satisfies the integral identity 
\beq
\label{uimv}
\frac{1}{T} \int_0^{T} u_i(t) \, {\rm d}t = \frac{1}{T} \int_0^{T} \varphi(t) \, {\rm d}t.
\eeq
\end{lemma}

\begin{lemma}
\label{Lemma2}
Let

\beq
\label{LambdaTh2}
\Lambda_i = \sqrt{\beta_i} + \frac{\sqrt{\gamma}}{T} \, \int_0^{T} \left(u_i(s) - \varphi(s)\right)^2 \, {\rm d}s \quad \text{for} \quad i \in \left\{H,L\right\},
\eeq

and
\beq
\label{QiTh2}
Q_i(t) = \gamma \, g(t) - \sqrt{\gamma \beta_i} - \gamma \left(u_i(t) - \varphi(t)\right)^2
\eeq
with $u_i(t)$ given by~\eqref{uiTh2}. Under assumptions~\eqref{R}, \eqref{pn2ori}, \eqref{gh}, \eqref{aST} and the additional assumption
$$
\Lambda_i < \frac{\sqrt{\gamma}}{T}  \int_0^{T} g(t) \, {\rm d}t,
$$
the unique real non-negative $T$-periodic solution of the problem 
\begin{equation}
\label{wieq}
\begin{cases}
w_i'(t) = \left(Q_i(t) - d \, w_i(t)\right) w_i(t), \quad \text{for } t \in (0,T),
\\
w_i(0) = w_i(T),
\end{cases}
\end{equation}
is
\beq
\label{wiTh2}
w_i(t) = \frac{\displaystyle d^{-1}\, \exp\left( \int_0^t Q_i(s) \, {\rm d}s \, \right)}{\displaystyle \frac{\displaystyle \int_0^T \exp\left(\int_0^s Q_i(z) \, {\rm d}z\right) {\rm d}s}{\displaystyle  \exp\left( \int_0^T Q_i(s) \, {\rm d}s \right) - 1} + \int_0^t \exp\left( \int_0^s Q_i(z) \, {\rm d}z \right) {\rm d}s} 
\eeq
and satisfies the integral identity 
\beq
\label{wimv}
\frac{1}{T} \int_0^{T} w_i(t) \, {\rm d}t = \frac{\sqrt{\gamma}}{d} \left(\frac{\sqrt{\gamma}}{T}  \int_0^{T} g(t) \, {\rm d}s - \Lambda_i \right).
\eeq
\end{lemma}
We refer the interested reader to~\cite{lorenzi2015dissecting} for the proofs of Lemma~\ref{Lemma1} and Lemma~\ref{Lemma2}.

\begin{theorem} \label{theorem2}
Under assumptions~\eqref{assbeta}, \eqref{R}, \eqref{pn2ori}, \eqref{gh} and the additional assumptions \eqref{aST}, the solution of the system of PDEs \eqref{eS1} subject to the initial condition~\eqref{eS1S0ic} is of the Gaussian form \eqref{gaussian} and satisfies the following:
\begin{itemize}
\item[(i)] if
\beq
\label{Th2ass0}
\min \left\{\Lambda_H, \Lambda_L\right\} \; \geq \; \frac{\sqrt{\gamma}}{T} \int_{0}^{T} g(t) \, {\rm d}t 
\eeq
then
\beq
\label{Th20}
\lim_{t\to\infty}\rho_H(t)=0  \quad \text{and} \quad \lim_{t\to\infty}\rho_L(t)=0;
\eeq
\item[(ii)] if 
\beq
\label{Th2assi}
\min \left\{\Lambda_H, \Lambda_L\right\} \; < \; \frac{\sqrt{\gamma}}{T} \int_{0}^{T} g(t) \, {\rm d}t, 
\eeq
and
$$
i = \argmin{k \in \left\{H,L \right\}}  \Lambda_k, \quad j = \argmax{k \in \left\{H,L \right\}} \Lambda_k, 
$$
then
\beq
\label{Th2i}
\rho_i(t) \to w_i(t), \quad \rho_j(t) \to 0 \quad \text{as } t \to \infty,
\eeq
and
\beq
\label{Th2imu}
\mu_i(t) \to u_i(t), \quad \sigma^2_i(t) \to \sqrt{\frac{\beta_i}{\gamma}} \quad \text{as } t \to \infty,
\eeq
with $w_i(t)$ and $u_i(t)$ given by~\eqref{wiTh2} and \eqref{uiTh2}, respectively.
\end{itemize}
\end{theorem}

\begin{proof}
Proposition~\ref{Prop1} ensures that the population density function $n_i(x,t)$ is of the Gaussian form \eqref{gaussian} with the population size, $\rho_i(t)$, the mean phenotypic state, $\mu_i(t)$, and the inverse of the related variance, $v_i(t)=1/\sigma_i^2(t)$, being governed by the Cauchy problem~\eqref{odes}. In this framework, we prove Theorem~\ref{theorem2} in 4 steps. In Step 1, we study the asymptotic behaviour of $v_i(t)$, $\mu_i(t)$ and $F_i(t)$ for $t \to \infty$. In Step 2, we show that $\rho_i(t)$ is non-negative and uniformly bounded. In Step 3, we prove claim~\eqref{Th20}. Finally, we prove claims~\eqref{Th2i} and \eqref{Th2imu} in Step 4.

\paragraph{Step 1: asymptotic behaviour of $v_i(t)$, $\mu_i(t)$ and $F_i(t)$ for $t \to \infty$.} Since the differential equation~\eqref{odes}$_1$ does not depend on $S(t)$, the expression~\eqref{vi} of $v_i(t)$ obtained in the proof of Theorem~\ref{theorem1} still holds and 
\begin{equation}
\label{viST}
v_i \to  \sqrt{\frac{\gamma}{\beta_i}} \quad \text{exponentially fast} \text{ as } t \to \infty.
\end{equation}
Moreover, using the asymptotic result~\eqref{viST} along with the linear differential equation~\eqref{odes}$_2$ for $\mu_i(t)$ one can easily show that 
\beq
\label{muiinfST}
\mu_i(t) \to u_i(t) \quad \text{exponentially fast} \text{ as } t \to \infty
\eeq
where $u_i(t)$ is a $T$-periodic solution of the differential equation~\eqref{uieq}. Lemma~\ref{Lemma1} ensures that $u_i(t)$ is given by~\eqref{uiTh2}. Lastly, for $F_i(t)$ defined according to~\eqref{Fit}, the asymptotic results~\eqref{viST} and~\eqref{muiinfST} allow us to conclude that
\begin{equation}
\label{FiinfST}
 F_i(t) \to \gamma \, g(t) - \sqrt{\gamma \, \beta_i}  - \gamma \left(u_i(t) - \varphi(t) \right)^2 \quad \text{exponentially fast} \text{ as } t \to \infty.
\end{equation}

\paragraph{Step 2: non-negativity and boundedness of $\rho_i(t)$.} Proceeding in a similar way as in the proof of Theorem~\ref{theorem1} (\emph{cf.} Step 2 in the proof of Theorem~\ref{theorem1}), one can prove that
\beq
\label{UBi}
0 \leq \rho_i(t) \leq \max \left\{\rho_i^0, \frac{2 \, \gamma}{d} \right\} \quad \text{for all } t \geq 0.
\eeq

\paragraph{Step 3: proof of claim~\eqref{Th20}.} Solving the differential equation~\eqref{odes}$_3$ for $\rho_i$ and imposing the initial condition~\eqref{odes}$_4$ yields
\begin{equation}
\label{rhoiST}
\rho_i(t) = \rho^0_i \exp\left[\int_0^{t} \left(F_i(s) - d \rho(s) \right) {\rm d}s\right],
\end{equation}
with $F_i(t)$ defined according to \eqref{Fit}. Combining the asymptotic result~\eqref{FiinfST} with the expression~\eqref{rhoiST} for $\rho_i(t)$ gives
\begin{eqnarray}
\label{rhoiSTmT}
\rho_i(t) &\sim& C \, \rho^0_i \exp\Bigg[\gamma \, \int_{0}^{t} g(s) \, {\rm d}s -  \sqrt{\gamma \, \beta_i} \, t - \gamma \, \int_{0}^{t} \left(u_i(s) - \varphi(s) \right)^2 \, {\rm d}s  \nonumber \\
&& \hspace{1.5cm} - d \int_0^{t} \hspace{-0.1cm}\rho(s) \, {\rm d}s\Bigg] \quad \text{as } t \to \infty,
\end{eqnarray}
for some positive constant $C$. Hence, using the fact that the functions $g(t)$, $\varphi(t)$ and $u_i(t)$ are $T$-periodic and considering ${\rm m} \to \infty$, we find
\begin{eqnarray}
\label{rhoiSTmT2}
\rho_i(t) &\sim& C \exp\Bigg[\gamma \, {\rm m} \,  \int_{0}^{T} g(t) \, {\rm d}t -  {\rm m} \, T \sqrt{\gamma \, \beta_i} - \gamma \, {\rm m} \, \int_{0}^{T} \left(u_i(t) - \varphi(t) \right)^2 \, {\rm d}t  \nonumber \\
&& \hspace{1.5cm} - d \int_0^{t} \hspace{-0.1cm}\rho(s) \, {\rm d}s\Bigg] \quad \text{as } t \to \infty,
\end{eqnarray}
for some positive constant $C$. Since the function $\rho(t)$ is non-negative [\emph{cf.} the uniform lower bound~\eqref{UBi}], the asymptotic relation~\eqref{rhoiSTmT2} ensures that if
$$
\sqrt{\beta_i} + \frac{\sqrt{\gamma}}{T} \int_{0}^{T} \left(u_i(t) - \varphi(t) \right)^2 \, {\rm d}t  \geq \frac{\sqrt{\gamma}}{T} \int_{0}^{T} g(t) \, {\rm d}t
$$ 
then
\beq
\label{extinctionST}
\lim_{t\to\infty}\rho_i(t)=0.
\eeq
This proves that if assumption~\eqref{Th2ass0} is satisfied then claim~\eqref{Th20} is verified.

\paragraph{Step 4: proof of claims~\eqref{Th2i} and ~\eqref{Th2imu}.} Let 
\beq
\label{indne}
i = \argmin{k \in \left\{H,L \right\}}  \Lambda_k \quad \text{and} \quad j = \argmax{k \in \left\{H,L \right\}} \Lambda_k.
\eeq
As long as $\rho_j(t)>0$, we can compute the quotient of $\rho_i(t)$ and $\rho_j(t)$ through~\eqref{rhoiST}. In so doing, using the asymptotic relation~\eqref{rhoiSTmT} for $\rho_i(t)$ and $\rho_j(t)$ and considering ${\rm m} \to \infty$ we obtain
\begin{equation}
\label{quotne2}
\frac{\rho_i(t)}{\rho_j(t)} \sim C \exp\bigg[{\rm m} \, T \sqrt{\gamma} \, \Big(\Lambda_j - \Lambda_i \Big) \bigg] \quad \text{as } t \to \infty,
\end{equation}
for some positive constant $C$, with $\Lambda_i$ and $\Lambda_j$ defined according to~\eqref{LambdaTh2}. The asymptotic relation~\eqref{quotne2} allows us to conclude that
\beq
\label{limquotne}
\lim_{t \to \infty} \frac{\rho_i(t)}{\rho_j(t)} = \infty.
\eeq
Since $\rho_i$ is uniformly bounded from above [\emph{cf.} the uniform upper bound \eqref{UBi}], the asymptotic result~\eqref{limquotne} implies that
\beq
\label{rhoSTto0}
 \rho_j(t) \to 0  \quad \text{exponentially fast} \text{ as } t \to \infty.
\eeq
We can rewrite the differential equation~\eqref{odes}$_3$ for $\rho_i$ as
\begin{equation}
\label{oderhoiST}
\rho_i'(t) = \left[\gamma g(t) - \sqrt{\gamma \, \beta_i}  - \gamma \, \left(u_i(t) - \varphi(t) \right)^2 + \eta(t) - d \rho_i(t)\right] \rho_i(t),
\end{equation}
where the function $\eta(t)$ is defined as 
$$
\eta(t) = \left(\sqrt{\gamma \, \beta_i} - \frac{\gamma}{v_i(t)} \right) + \gamma \, \left[\left(u_i(t) - \varphi(t) \right)^2 - \left(\mu_i(t) - \varphi(t) \right)^2\right] - d \rho_j(t).
$$
Using the asymptotic results~\eqref{viST}, \eqref{muiinfST} and~\eqref{rhoSTto0} we see that $\eta(t)  \to 0$ exponentially fast as $t \to \infty$. Hence, $\rho_i(t) \to \tilde{\rho}_i(t) $ as $t \to \infty$, with $\tilde{\rho}_i(t)$ being the solution of the differential equation
\begin{equation}
\label{oderhoiSTtilde}
\tilde \rho_i' = f(\tilde \rho_i,t),
\end{equation}
with
$$
f(\tilde \rho_i,t) = \left[\gamma g(t) - \sqrt{\gamma \, \beta_i}  - \gamma \, \left(u_i(t) - \varphi(t) \right)^2 - d \tilde \rho_i\right] \tilde \rho_i,
$$
subject to an initial condition $0<\tilde \rho_i(0)<\infty$. We note that: (i) the function $\tilde \rho_i(t)$ is uniformly bounded as it satisfies the upper and lower bounds 
$$
0 \leq \tilde \rho_i(t) \leq \max \left\{\tilde \rho_i(0), \frac{2 \, \gamma}{d} \right\} \quad \text{for all } t \geq 0;
$$
(ii) the function $f$ is Lipschitz continuous in the first variable; (iii) the function $f$ is Lipschitz continuous and $T$-periodic in the second variable, since $g$, $\varphi$ and $u_i$ are $T$-periodic Lipschitz continuous functions of $t$ [\emph{cf.} assumptions \eqref{gphiSTa} and expression \eqref{uiTh2}]. Therefore, the conditions of Massera's Convergence Theorem~\cite{massera1950existence,smith1986massera} are satisfied and this allows us to conclude that
\beq
\label{tilderhoiinfST}
\tilde \rho_i(t) \longrightarrow w_i(t) \quad \text{as } t \to \infty,
\eeq
with $w_i(t)$ being a non-negative $T$-periodic solution of the differential equation~\eqref{wieq}. Under the additional assumption \eqref{Th2assi}, that is,
$$
\Lambda_i < \; \frac{\sqrt{\gamma}}{T} \int_{0}^{T} g(t) \, {\rm d}t,
$$
Lemma~\ref{Lemma2} ensures that $w_i(t)$ is given by~\eqref{wiTh2}. Claims~\eqref{Th2i} and \eqref{Th2imu} follow from the asymptotic results~\eqref{rhoSTto0} and \eqref{tilderhoiinfST} along with the asymptotic result~\eqref{muiinfST} and~\eqref{viST} with $i = \argmin{k \in \left\{H,L \right\}}  \Lambda_k$. 
\qed
\end{proof}

In summary, Theorem~\ref{theorem2} gives an explicit characterisation of the long-term limit of $v_i(t)$, $\mu_i(t)$ and $\rho_i(t)$ for the surviving population $i$ and shows that the surviving population is the one characterised by the larger positive value of the quantity
$$
\frac{1}{T} \int_0^T F^{\infty}_i(t) \; {\rm d}t = \frac{\gamma}{T} \int_0^T g(t) \; {\rm d}t  - \sqrt{\gamma} \, \Lambda_i
$$
where 
$$
F^{\infty}_i(t) = \lim_{t \to \infty} F_i(t) = \gamma \, g(t) - \sqrt{\gamma \beta_i} - \gamma \, \left(u_i(t) - \varphi(t)\right)^2.
$$

\begin{remark}
\label{periodoscil}
Since the functions $u_i(t)$ and $w_i(t)$ are $T$-periodic and satisfy the integral identities \eqref{uimv} and \eqref{wimv}, respectively, the results established by Theorem~\ref{theorem2} show that the long-term limits of the size and the mean phenotypic state of the surviving population are periodic functions of time with period $T$ and mean values given by~\eqref{uimv} and \eqref{wimv}, respectively.
\end{remark}

\begin{remark}
\label{remarkLambda}
Using the differential equation~\eqref{uieq} for $u_i(t)$ one can easily obtain
$$
\frac{{\rm d}}{{\rm d}t} \left(u_i - \varphi \right)^2 = 4 \, \gamma \, \sqrt{\gamma \, \beta_i} \left[\frac{1}{2} \, \frac{1}{\sqrt{\gamma \, \beta_i}} \, \left(\varphi - u_i\right) \varphi' -  \left(u_i - \varphi \right)^2\right].
$$
Integrating both sides of the above equation with respect to $t$ between $0$ and $T$, and using the fact that $u_i(T) - \varphi(T) = u_i(0) - \varphi(0)$, yields
\beq
\label{uimphisqmv}
\frac{1}{T} \, \int_0^{T} \left(u_i(t) - \varphi(t) \right)^2 \, {\rm d}t = \frac{1}{2 \sqrt{\gamma \, \beta_i}} \, \frac{1}{T} \, \int_0^{T} \left(\varphi(t) - u_i(t)\right) \varphi'(t)  \, {\rm d}t.
\eeq
Therefore, definition~\eqref{LambdaTh2} can be rewritten as
\beq
\label{LambdaTh2altern}
\Lambda_i = \sqrt{\beta_i} +  \frac{1}{2 \sqrt{\beta_i}} \, \frac{1}{T} \, \int_0^{T} \left(\varphi(t) - u_i(t)\right) \varphi'(t)  \, {\rm d}t.
\eeq
If $S \equiv \bar{S}$ then $g(t) \equiv \overline{g}$ and $\varphi(t) \equiv \overline{\varphi}$ (\emph{i.e.} $\varphi' \equiv 0$). In this case,
$$
\frac{\sqrt{\gamma}}{T}\int_0^T g(t) \, {\rm d}t = \overline{g}
$$
and~\eqref{LambdaTh2altern} allows one to see that $\Lambda_i=\sqrt{\beta_i}$. Hence, if $S$ is constant then the results of Theorem~\ref{theorem2} reduce to the results of Theorem~\ref{theorem1}. Moreover, the first term in the expression~\eqref{LambdaTh2altern} for $\Lambda_i$ is clearly a monotonically increasing function of $\beta_i$, whereas the factor in front of the integral in the second term is a monotonically decreasing function of $\beta_i$. Hence, if the mean value of the $T-$periodic function $\left(\varphi(t) - u_i(t)\right) \varphi'(t)$ is sufficiently small then $\Lambda_H > \Lambda_L$, while if such a mean value is sufficiently large then $\Lambda_H < \Lambda_L$. We expect the latter scenario to occur when the variability and the rate of change of $S(t)$ are sufficiently high so as to cause substantial and sufficiently fast variations in the value of $\varphi(t)$.
\end{remark}

The asymptotic results established by Theorem~\ref{theorem2} formalise mathematically the idea that in periodically fluctuating environments:
\begin{enumerate}
\item populations undergoing spontaneous phenotypic variations at a rate too large compared to the mean value of the maximum fitness will ultimately go extinct [\emph{cf.} point (i)];
\item {\it ceteris paribus}, if at least one population undergoes spontaneous phenotypic variations at a rate sufficiently small compared to the mean value of the maximum fitness, then the following behaviours are possible: 
\begin{enumerate}
\item[2(a).] when environmental conditions are relatively stable, the population with the lower rate of phenotypic variations will outcompete the other population [\emph{cf.} point (ii) and Remark~\ref{remarkLambda}];
\item[2(b).] when environmental conditions undergo drastic changes, either both populations go extinct  [\emph{cf.} point (i) and Remark~\ref{remarkLambda}] or the population with the higher rate of phenotypic variations will outcompete the other population [\emph{cf.} point (ii) and Remark~\ref{remarkLambda}];
\item[2(c).] the phenotype distribution of the surviving population will be unimodal, and both the population size and the mean phenotype will become periodic [\emph{cf.} point (ii) and Remark~\ref{periodoscil}];
\item[2(d).] ultimately, the population size and the mean phenotype will both oscillate with the same period as the fluctuating environment, and the mean value (with respect to time) of the mean phenotype will be the same as the mean value of the fittest phenotypic state with the related variance being directly proportional to the rate of phenotypic variations [\emph{cf.} point (ii) and Remark~\ref{periodoscil}].
\end{enumerate} 
\end{enumerate} 
These biological implications are reinforced by the numerical solutions presented in the next section.

\section{Numerical simulations} \label{Section4}
In this section, we construct numerical solutions to the system of non-local parabolic PDEs~\eqref{eS1} subject to the initial condition~\eqref{eS1S0ic}. In Section~\ref{Section41}, we describe the numerical methods employed and the set-up of numerical simulations. In Section \ref{Section42}, we present a sample of numerical solutions that confirm the results of our analysis of evolutionary dynamics, both in the case where $S(t)$ is constant and when $S(t)$ oscillates periodically. 

\subsection{Numerical methods and set-up of numerical simulations} \label{Section41}
We select a uniform discretisation consisting of $2000$ points on the interval $[-5,5]$ as the computational domain of the independent variable $x$ and impose no flux boundary conditions. Moreover, we assume $t \in [0,t_{f}]$, with $t_f>0$ being the final time of simulations, and we discretise the interval $[0,t_{f}]$ with the uniform step $\Delta t=0.0001$. The method for constructing numerical solutions to the system of non-local parabolic PDEs~\eqref{eS1} is based on an explicit finite difference scheme in which a three-point stencil is used to approximate the diffusion terms and an explicit finite difference scheme is used for the reaction term~\cite{leveque2007finite}. On the other hand, we use the {\sc Matlab} built-in solver {\sc ode45} to solve numerically the Cauchy problem~\eqref{odes} for $v_i(t)$, $\mu_i(t)$ and $\rho_i(t)$.

The parameter values used to carry out numerical simulations are listed in Table~\ref{table1}. In summary, to capture the fact that rates of spontaneous phenotypic variation are small, in general, and much smaller than maximum proliferation rates, in particular, we assume $\beta_i \ll \gamma$ for $i \in \left\{H,L\right\}$. Furthermore, given the values of $\gamma$ and $\beta_i$, we fix the value of $d$ to be such that the long-term limit~\eqref{Th1ii} of the size of the population $L$ is approximatively $10^4$, which is consistent with biological data from the existing literature regarding {\it in vitro} cell populations~\cite{voorde2019improving}. 

We remark that the value of the parameter $\beta_L$ and the range of values of the parameter $\beta_H$ reported in Table~\ref{table1} are such that neither condition~\eqref{Th1assi} nor condition~\eqref{Th2ass0} are met in all cases on which we report here. This ensures that the two populations do not simultaneously go extinct.
\begin{table}[h!] \centering 
\begin{tabular}{|l|l|l|}
\hline
\textbf{Parameter} & \textbf{Description}      & \textbf{Value / Value range} \\ \hline
$\gamma$     & Maximum proliferation rate                    &  100              \\ 
$d$             & Death rate due to competition            &      0.01          \\
$\beta_L$ & Rate of phenotypic variation of population $L$ & 0.01 \\ 
$\beta_H$ & Rate of phenotypic variation of population $H$ & [0.01, 0.1] \\
 \hline
\end{tabular}
\caption{Parameter values used to carry out numerical simulations.}
\label{table1}
\end{table}

We consider both populations to have the same initial phenotypic distribution~\eqref{eS1S0ic} with $v_i^0 = 20$, $\mu_i^0 = 0$ and $\rho_i^0 \approx 800$ for $i \in \left\{H,L\right\}$.

\subsection{Main results} \label{Section42}
We consider the following definition of the nutrient concentration
\begin{equation} \label{sinS}
S(t) = M + A \sin\left({\frac{2\pi t}{T}}\right).
\end{equation}
In definition~\eqref{sinS}, the parameter $M > 0$ represents the mean nutrient concentration, while the parameter $A \geq 0$ models the semi-amplitude of the oscillations of the nutrient concentration, which have period $T>0$. Clearly, we consider only values of $M$ and $A$ such that $S(t) \geq 0$, \emph{i.e.} $0 \leq A \leq M$.

We start by exploring three prototypical scenarios exemplified by different values of the parameter $A$. In particular, we choose $M=1$ and compare the numerical solutions obtained for $A=0$ (\emph{i.e.} constant nutrient concentration), $A=0.5$ (\emph{i.e.} lower nutrient variability) and $A=1$ (\emph{i.e.} higher nutrient variability). Figure~\ref{fig2a} displays plots of the nutrient concentration $S(t)$, the rescaled maximum fitness $g(t)$ and the fittest phenotypic state $\varphi(t)$ corresponding to such choices of the parameter $A$. These plots show that, as one would expect, higher nutrient variability brings about more pronounced variations in the rescaled maximum fitness and the fittest phenotypic state.
\begin{figure}[h!] \centering
\includegraphics[width=1 \linewidth]{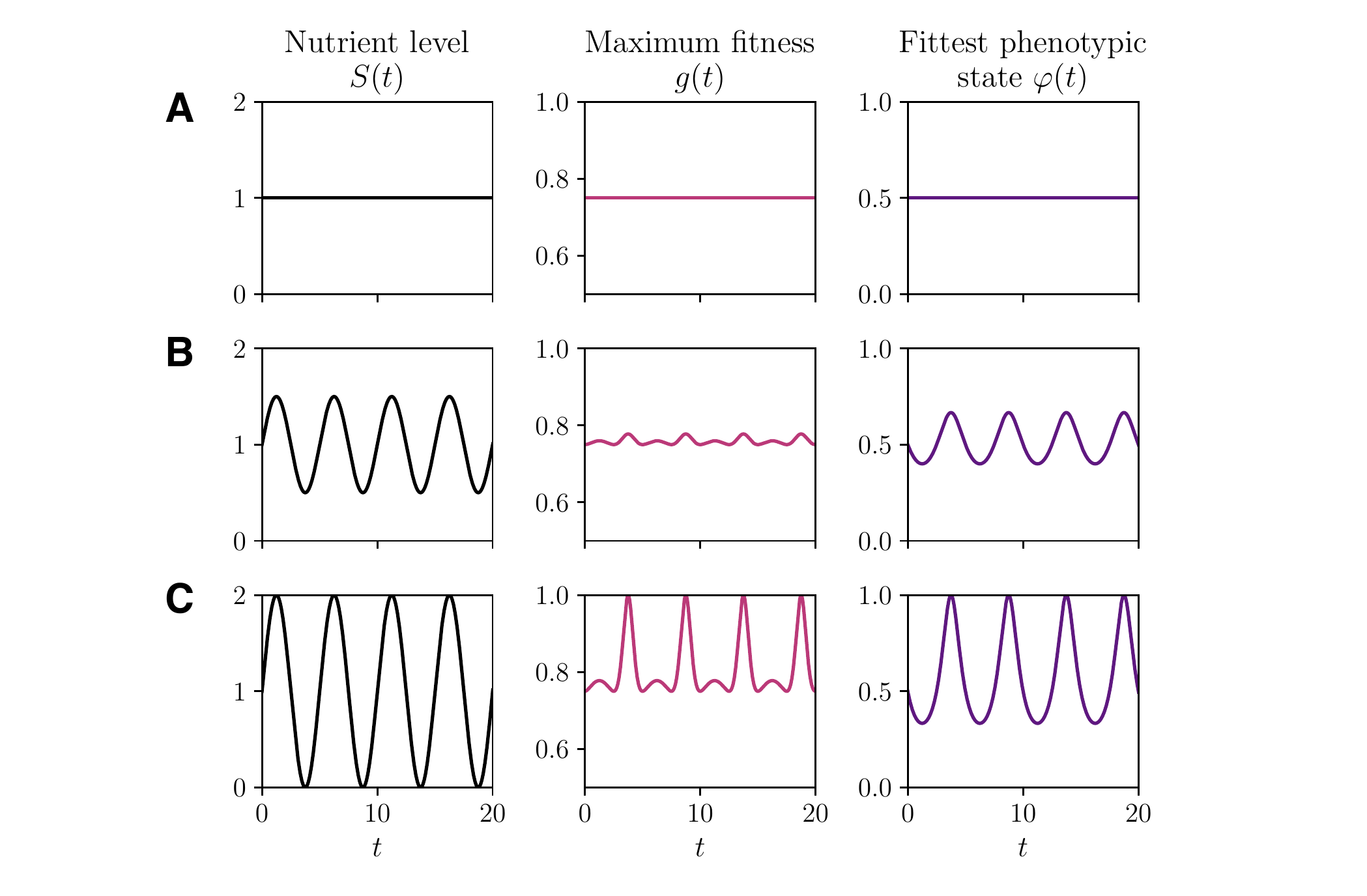}
\caption{\label{fig2a} {\bf {\textsf A.}} Plots of the nutrient concentration $S(t)$ (left panel) defined according to~\eqref{sinS} with $M=1$ and $A=0$, and the corresponding rescaled maximum fitness $g(t)$ (central panel) and fittest phenotypic state $\varphi(t)$ (right panel) defined according to~\eqref{gh}. {\bf {\textsf B.}} Same as row {\bf {\textsf A}} but with $A=0.5$ and $T=5$. {\bf{\textsf C.}} Same as row {\bf {\textsf A}} but with $A=1$ and $T=5$}
\end{figure}

Figure \ref{fig2b} shows a comparison between the exact solutions~\eqref{gaussian} -- with $v_i(t)$, $\mu_i(t)$ and $\rho_i(t)$ obtained by solving numerically the Cauchy problem~\eqref{odes} -- and the numerical solutions of the system of non-local parabolic PDEs~\eqref{eS1} subject to the initial condition~\eqref{eS1S0ic}. In agreement with the results established by Proposition~\ref{Prop1}, for all values of $A$ considered, there is a perfect match between the population sizes obtained by computing numerically the integrals of the components of the numerical solution of the system of PDEs~\eqref{eS1} (\emph{cf.} solid lines in the left column of Figure \ref{fig2b}) and the population sizes obtained by solving numerically the Cauchy problem~\eqref{odes} (\emph{cf.} dashed and dotted lines in the left column of Figure \ref{fig2b}). Similarly, there is excellent agreement between the population density functions obtained by solving numerically the system of PDEs~\eqref{eS1} (\emph{cf.} solid lines in the right column of Figure \ref{fig2b}) and the population density functions~\eqref{gaussian} with $v_i(t)$, $\mu_i(t)$ and $\rho_i(t)$ given by the numerical solutions of the Cauchy problem~\eqref{odes} (\emph{cf.} dashed and dotted lines in the right column of Figure \ref{fig2b}).
\begin{figure}[htb!] \centering
\includegraphics[width=1 \linewidth]{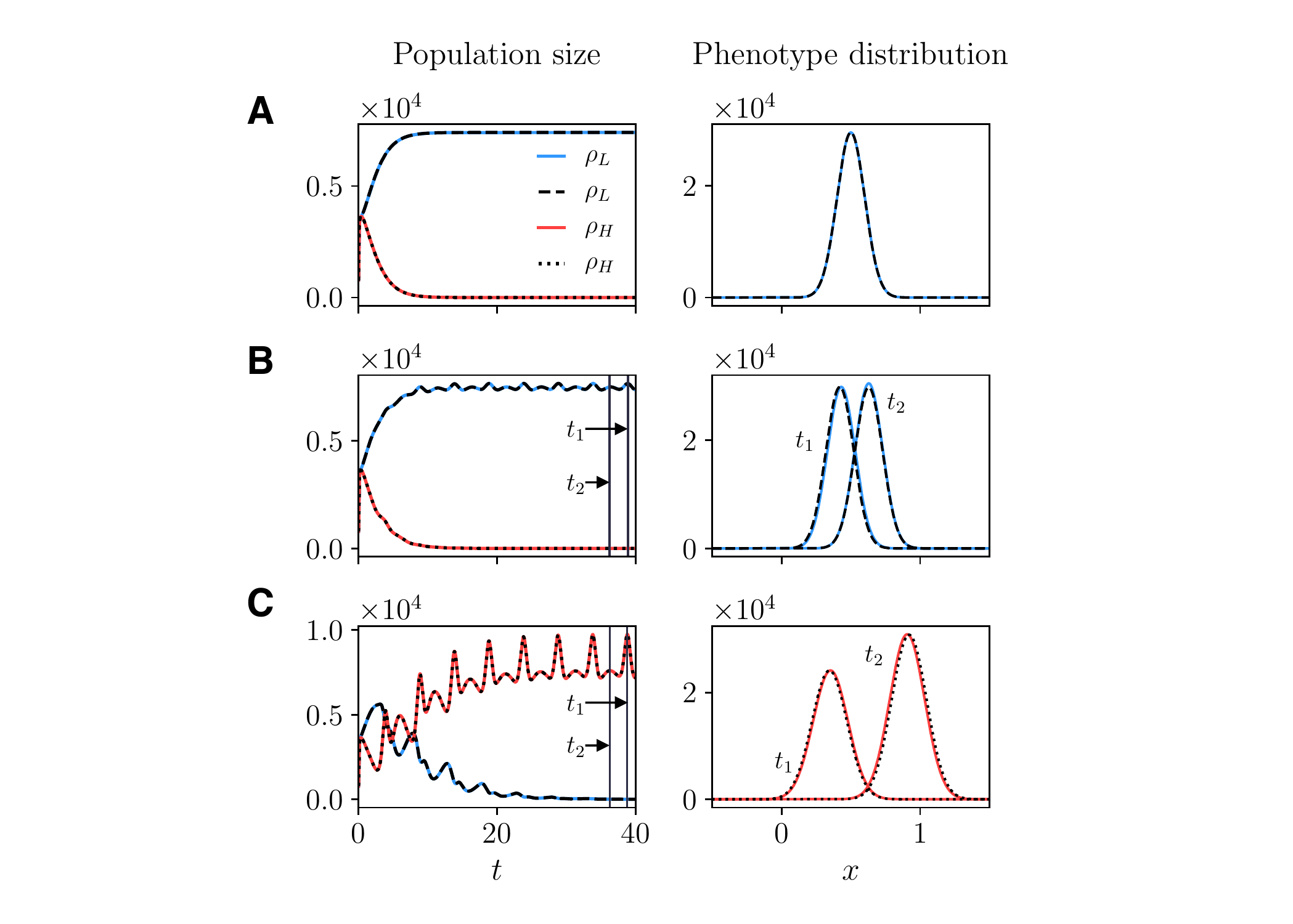}
\caption{\label{fig2b} {\bf Left column.} Plots of the population sizes $\rho_H(t)$ (red line) and $\rho_L(t)$ (blue line) obtained by computing numerically the integrals of the components of the numerical solution of the system of PDEs~\eqref{eS1} subject to the initial condition~\eqref{eS1S0ic}. The dotted and dashed lines highlight, respectively, $\rho_H(t)$ and $\rho_L(t)$ obtained by solving numerically the Cauchy problem~\eqref{odes}. The nutrient concentration $S(t)$ is defined according to~\eqref{sinS} with $M=1$ and $A=0$ (row {\bf {\textsf A}}), $M=1$, $A=0.5$ and $T=5$ (row {\bf {\textsf B}}), or $M=1$, $A=1$ and $T=5$ (row {\bf {\textsf C}}) -- \emph{cf.} the plots displayed in Figure~\ref{fig2a}. The values of the model parameters are those reported in Table~\ref{table1} with $\beta_H=0.025$. {\bf Right column.} Plots of the corresponding phenotype distribution of the surviving population obtained by solving numerically the system of PDEs~\eqref{eS1} subject to the the initial condition~\eqref{eS1S0ic}. In particular, the plot of $n_L(x,t_f)$ is shown in row {\bf {\textsf A}}, while the plots in rows {\bf {\textsf B}} and {\bf {\textsf C}} are, respectively, those of $n_L(x,t)$ and $n_H(x,t)$ at $t=t_1$ and $t=t_2$, with $t_1$ and $t_2$ being highlighted in the corresponding plots in the left column. The dotted and dashed lines correspond to the exact phenotype distributions~\eqref{gaussian} with $v_i(t)$, $\mu_i(t)$ and $\rho_i(t)$ given by the numerical solutions of the Cauchy problem~\eqref{odes}} 
\end{figure}

In accord with the results of Theorem~\ref{theorem1}, when the nutrient concentration is constant (\emph{i.e.} $S(t) \equiv M$), the population with the lower rate of phenotypic variations (\emph{i.e.} population $L$) outcompetes the other population ({\it vid.} Figure~\ref{fig2b}{\textsf A}). The size of the surviving population $\rho_L(t)$ reaches the asymptotic value~\eqref{Th1ii} and the phenotype distribution at the end of the simulations $n_L(x,t_f)$ is Gaussian with mean and variance equal to the asymptotic values~\eqref{Th1iii}. 

In agreement with the results established by Theorem~\ref{theorem2} (\emph{vid.} Remark~\ref{remarkLambda}), a similar outcome is observed in the presence of a low nutrient variability ({\it vid.} Figure~\ref{fig2b}{\textsf B}). In fact, in this case $\Lambda_L < \Lambda_H$. On the contrary, the population with the higher rate of phenotypic variations (\emph{i.e.} population $H$) outcompetes the other population when the nutrient variability is sufficiently high ({\it vid.} Figure~\ref{fig2b}{\textsf C}). This is due to the fact that in this case the condition $\Lambda_H < \Lambda_L$ is met. As expected (\emph{cf.} Remark~\ref{periodoscil}), since $A>0$ both the size and the mean phenotype of the surviving population become $T$-periodic, with mean values given by~\eqref{uimv} and \eqref{wimv}, respectively. Moreover, the phenotype distribution of the surviving population remains Gaussian with variance given by~\eqref{Th2imu}. This implies that if population $H$ outcompetes population $L$ then the variance of the phenotype distribution (\emph{i.e.} the level of phenotypic heterogeneity) will be ultimately larger than in the case where population $L$ is selected. 

Taken together, these results demonstrate that when the nutrient concentration is constant, or in the presence of a low level of nutrient variability, it is evolutionarily more desirable to rarely undergo spontaneous phenotypic variations, since environmental conditions are stable. Conversely, when nutrient variability is high (\emph{i.e.} alternating cycles of starvation and nutrient abundance occur), higher rates of spontaneous phenotypic variations constitute a competitive advantage, as they allow for a quicker adaptation to changeable environmental conditions, and higher level of phenotypic heterogeneity emerge.

Exploiting the results of the analysis of evolutionary dynamics developed in Section~\ref{Section3}, we can further assess the range of environmental conditions under which higher rates of spontaneous phenotypic variations will represent a source of competitive advantage. In more detail, as shown by the asymptotic results established by Theorem~\ref{theorem2}, provided that condition~\eqref{Th2ass0} is not satisfied (\emph{i.e.} at least one population survives), the outcome of competition between population $H$ and population $L$ in periodically fluctuating environments can be predicted by computing the value of the quantities $\Lambda_H$ and $\Lambda_L$ given by~\eqref{LambdaTh2}. In particular, the population characterised by the lower value of this quantity will ultimately be selected. Therefore, we computed $\Lambda_H$ and $\Lambda_L$ for different values of the period of the nutrient oscillations $T$ and different values of the rate of spontaneous phenotypic variations $\beta_H$. We used the values of the other evolutionary parameters reported in Table~\ref{table1} and considered possible values of the environmental parameters $M$ and $A$ corresponding to three different scenarios: an environment whereby the nutrient is abundant and undergoes small-amplitude periodic oscillations, \emph{i.e.} $M$ is relatively large and $A$ is relatively small (\emph{vid.} Figure~\ref{fig4}{\textsf A}); an environment whereby the nutrient is scarce and undergoes small-amplitude periodic oscillations, \emph{i.e.} $M$ and $A$ are both relatively small (\emph{vid.} Figure~\ref{fig4}{\textsf B}); and an environment whereby periodic oscillations can induce a sufficiently high variability of nutrient concentration, \emph{i.e.} $M=A$ and different values of $A$ are allowed (\emph{vid.} Figure~\ref{fig4}{\textsf C}). 
\begin{figure}[h!] \centering
\includegraphics[width=1 \linewidth]{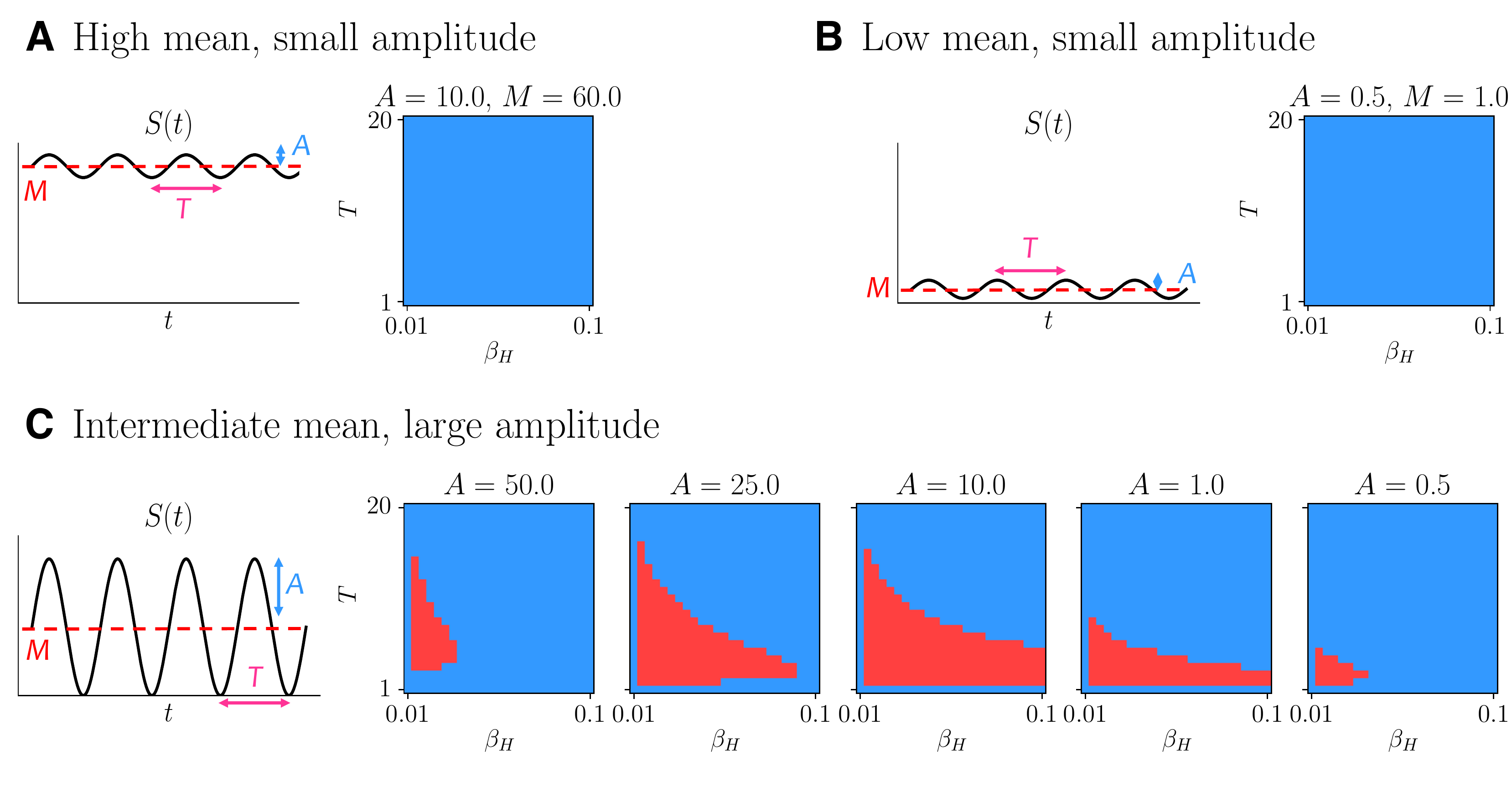}
\caption{\label{fig4} {\bf {\textsf A.}} Qualitative dynamics of the nutrient concentration $S(t)$ defined according to~\eqref{sinS} with $M=60$ and $A=10$, and corresponding plot of $\sgn \left(\Lambda_H -  \Lambda_L\right)$ as a function of $\beta_H~\in~(\beta_L, 0.1]$, with $\beta_L=0.01$, and $T \in [1,20]$. The quantities $\Lambda_H$ and $\Lambda_L$ are computed using~\eqref{LambdaTh2} and the values of the other model parameters reported in Table~\ref{table1}. The blue points in the $\beta_H-T$ plane correspond to $\sgn \left(\Lambda_H -  \Lambda_L\right) =1$ (\emph{i.e.} $\Lambda_L < \Lambda_H$), whereas the red points correspond to $\sgn \left(\Lambda_H -  \Lambda_L\right) = -1$ (\emph{i.e.} $\Lambda_H < \Lambda_L$). {\bf {\textsf B.}} Same as panel {\bf {\textsf A}} but with $M=1$ and $A=0.5$. {\bf {\textsf C.}} Same as panel {\bf {\textsf B}} but with $M=A$ and $A \in \left\{0.5, 1, 10, 25, 50\right\}$}
\end{figure}

The results obtained are summarised by the plots in Figure~\ref{fig4}. As we would expect (\emph{cf.} Remark~\ref{remarkLambda}), if the nutrient concentration undergoes low-amplitude periodic oscillations then $\Lambda_L < \Lambda_H$ for all values of $T$ and $\beta_H$ considered (\emph{vid.} Figure~\ref{fig4}{\textsf A} and Figure~\ref{fig4}{\textsf B}). On the other hand, when periodic oscillations can bring about sufficiently high levels of nutrient variability, there is a region of the $\beta_H-T$ plane where $\Lambda_H <  \Lambda_L$ (\emph{vid.} Figure~\ref{fig4}{\textsf C}). When the value of $A$ is either low or high, this region is small and concentrated in the bottom left corner of the plane. For intermediate values of $A$ the region where $\Lambda_H <  \Lambda_L$ is wider and such that the smaller the value $T$ (\emph{i.e.} the higher the frequency of the nutrient oscillations) the wider the range of values of $\beta_H$ that belong to it. 

Furthermore, we can investigate how the fluctuations in the nutrient level affect the phenotypic distribution of each population at any given time point by constructing numerical solutions to the system of non-local PDEs \eqref{erhos} subject to initial condition \eqref{eS1S0ic} with $S(t)$ defined according to \eqref{sinS}. This is demonstrated by the plots in Figure \ref{fig9}, which show sample dynamics of the nutrient concentration $S(t)$, the phenotype distributions $n_H(x,t)$ and $n_L(x,t)$, and the population sizes $\rho_H(x,t)$ and $\rho_L(x,t)$, for different values of semi-amplitude $A$ and mean $M$ of the fluctuations.
\begin{figure}[h!] \centering
\includegraphics[width=1 \linewidth]{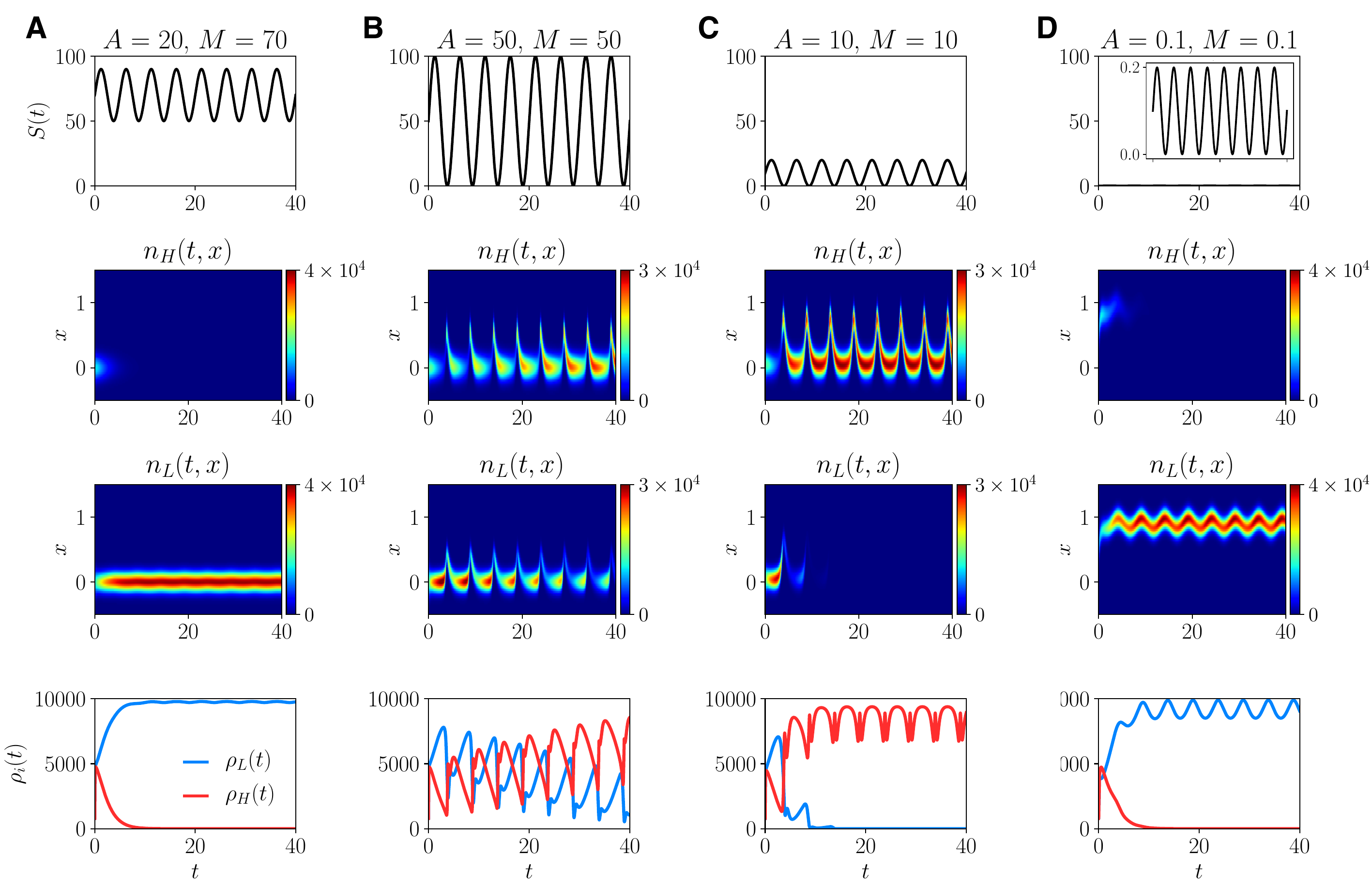}
\caption{\label{fig9} {\bf {\textsf A.}} Plots of the nutrient concentration $S(t)$ (first row), the phenotype distributions $n_H(x,t)$ (second row) and $n_L(x,t)$ (third row), and the population sizes $\rho_H(t)$ (fourth row, red line)  and $\rho_L(t)$ (fourth row, blue line) obtained by solving numerically the system of PDEs~\eqref{eS1}, where $S(t)$ is defined according to \eqref{sinS} with $T=5$, $A=20$ and $M=70$. The values of the model parameters are defined as in Table~\ref{table1} with $\beta_H = 0.025$. {\bf {\textsf B} - {\textsf D.}} Same as column {\bf {\textsf A}}  but for $A=M=50$ (column {\bf {\textsf B}}), $A=M=10$ (column {\bf {\textsf C}}), and $A=M=0.1$ (column {\bf {\textsf D}})}
\end{figure}

When the nutrient is abundant and experiences fluctuations of relatively low level (\emph{vid.} Figure~\ref{fig9}{\textsf A}) population $L$ outcompetes population $H$. This is due to the fact that, as shown by Figure~\ref{figS1}{\textsf A} in Appendix {\textsf A}, the fittest phenotypic state $\varphi(t)$ undergoes very small periodic oscillations and its value remains close to $0$ (\emph{i.e.} the value of the phenotypic variable $x$ corresponding to the fittest phenotypic state when nutrient is abundant). Moreover, the rescaled maximum fitness $g(t)$ undergoes very small periodic oscillations and its value remains close to $1$. 

When the nutrient level is uniformly low and undergoes relatively small oscillations (\emph{vid.} Figure~\ref{fig9}{\textsf D}) population $L$ is selected against population $H$. This is due to the fact that, as shown by Figure~\ref{figS1}{\textsf D} in Appendix {\textsf A}, both the fittest phenotypic state $\varphi(t)$ and the rescaled maximum fitness $g(t)$ undergo small periodic oscillations and their values remain close to $1$. We recall that $x=1$ is the value of the phenotypic variable corresponding to the fittest phenotypic state when nutrient is scarce.

For the cases when the populations experience fluctuations of relatively high level (\emph{vid.} Figure~\ref{fig9}{\textsf B} and Figure~\ref{fig9}{\textsf C}) population $H$ outcompetes population $L$. This is due to the fact that, as shown by Figure~\ref{figS1}{\textsf B} and Figure~\ref{figS1}{\textsf C} in Appendix {\textsf A}, the fittest phenotypic state $\varphi(t)$ fluctuates periodically between $1$ and a positive value close to $0$. Moreover, the rescaled maximum fitness $g(t)$ undergoes small periodic oscillations and its value remains close to $1$.

Note that in all cases the phenotype distribution of the surviving population remains unimodal with maximum at the mean phenotypic state. Ultimately, both the size and the mean phenotypic state of the surviving population oscillate periodically with the same period as the nutrient concentration $S(t)$. Furthermore, when the populations experience fluctuations of relatively high level (\emph{i.e.} Figure~\ref{fig9}{\textsf B} and Figure~\ref{fig9}{\textsf C}), the mean phenotypic state of the surviving population (\emph{i.e.} population $H$) undergoes rapid transitions between $1$ and a positive value close to $0$. This can be biologically seen as the emergence of a bet-hedging behaviour.

\section{Interpretation of the results in the context of cancer metabolism} \label{Section5}
The generality of the model and the robustness of the results make our conclusions applicable to a broad range of asexual populations evolving in fluctuating environments. As an example, in this section we discuss the biological implications of our mathematical results in the context of cancer cell metabolism and tumour-microenvironment interactions.

Cancers begin from single cells that grow to form organ-like masses within multicellular organisms. A fundamental property of cancer cells is a self-defined fitness function such that their proliferation is determined by their heritable phenotypic properties and the local environmental selection forces. That is, individual cancer cells have the capacity to evolve novel phenotypes and to adapt to the often harsh intratumoural environment. In contrast, while normal epithelial cells have the capacity to change their phenotype to some degree, e.g. they can only do so within normal physiological constraints in response to stress. In other words, because the proliferation of normal cells is entirely governed by local tissue constraints, they cannot, unlike cancer cells, evolve adaptations to many non-physiological conditions.

This difference in evolutionary capacity (or reaction norm) can confer a significant adaptive advantage to cancer cells. For example, cancer cells often metabolise glucose using glycolytic (converting glucose to lactic acid) pathways even when the oxygen concentration is sufficient for the aerobic mechanism (converting glucose to H$_2$O and CO$_2$). Known as the Warburg effect~\cite{warburg1925metabolism}, this can be understood in a Darwinian context as a form of niche construction because inefficiency of ATP (energy currency of cells) production is offset  by the evolutionary  advantage of generating a locally acidic environment. Cancer cells can evolve adaptive strategies to survive and proliferate in such an environment but normal cells cannot.

Angiogenesis is another form of niche construction as cancer cells, acting as a loosely organised group, produce and excrete pro-vascular proteins such as VEGF. Importantly, angiogenesis in cancers will occur entirely through these local interactions so that new vascular sprouts will emerge from the nearest vessel regardless of its flow capacity. Furthermore, cancer cells, once receiving blood flow, have no evolutionary imperative to invest resources in vascular maturation that permits blood flow to other regions of the tumour~\cite{gillies2018eco}. Such an unregulated vascular network will be highly unstable with cycles of growth and regression and blood flow dynamics that are inevitably disordered.

A number of studies have demonstrated that this disordered process of angiogenesis produces stochastic variations in blood flow leading to cycles of perfusion, cessation of flow, and then re-perfusion~\cite{kimura1996fluctuations}. This produces corresponding fluctuations in local environmental conditions that are dependent on blood flow, including oxygen and glucose, retention of metabolites such as acid, and important signalling molecules such as testosterone or oestrogen. Particularly, regions of normoxia (normal levels of oxygen), chronic hypoxia (low oxygen level) and cycling hypoxia have been distinguished in experimental and clinical studies~\cite{michiels2016cycling}. If we assume that $S(t)$ in our model represents the oxygen level at time $t$, then the phenotypic variants best adapted to oxygen-rich environments, and thus displaying a regular metabolism, are those with $x \rightarrow 0$, while the phenotypic variants best adapted to oxygen-low environments -- \emph{i.e.} the phenotypic variants that proliferate through the consumption of glucose, which is usually abundant -- correspond to $x \rightarrow 1$. For simplicity, we ignore any costs associated with the choice of metabolic preference, so that assumption~\eqref{asimp} is satisfied. This is illustrated by the schemes in Figure~\ref{figx}. 
\begin{figure}[h] \centering
\includegraphics[width=1 \linewidth]{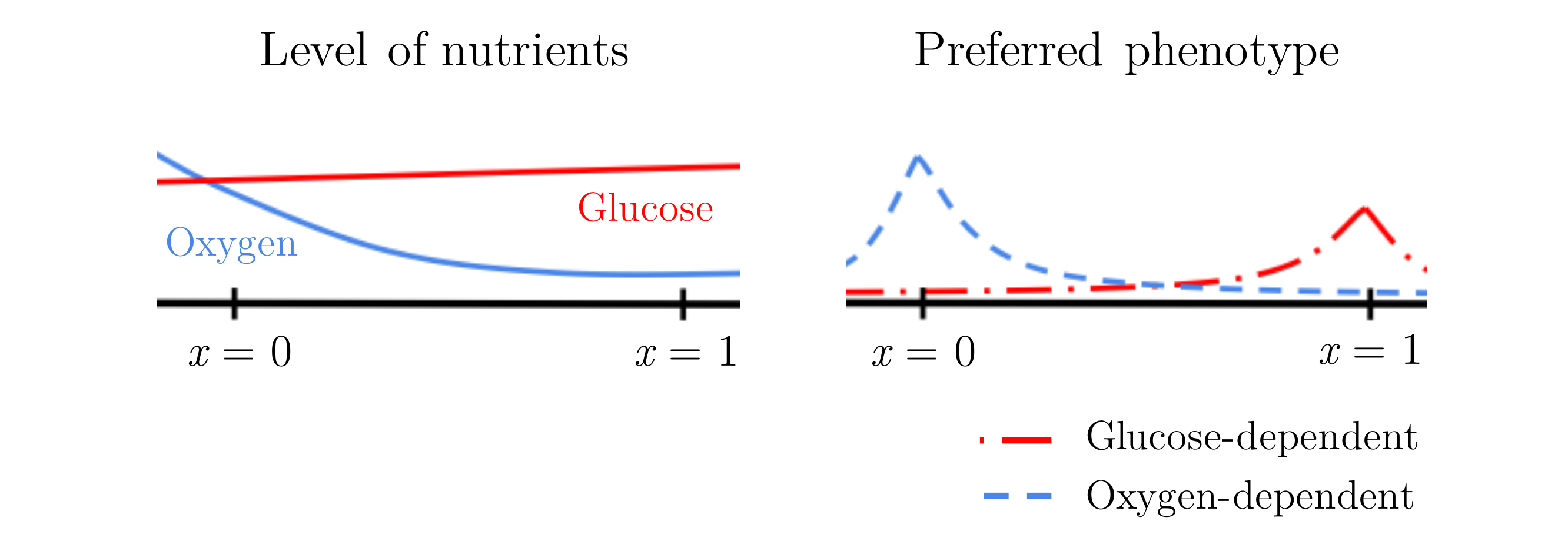}
\caption{\label{figx}Schematic interpretation of the phenotypic fitness landscape of our model in the context of cancer metabolism}
\end{figure}

Regions of normoxia and chronic hypoxia are analogous to cases {\textsf A} and {\textsf B} in Figure~\ref{fig4} where low levels of environmental variability are observed. Our results support the idea that these regions will be mainly populated by phenotypic variants best adapted to either oxygen-rich or oxygen-low environments. Moreover, since our results indicate that a higher rate of phenotypic variation does not constitute a competitive advantage in the presence of small environmental fluctuations, we expect relatively low levels of phenotypic heterogeneity to be observed in regions of either normoxia or chronic hypoxia.

On the contrary, regions of cyclic hypoxia are characterised by high variability in the oxygen levels. This can be related to case {\textsf C} in Figure~\ref{fig4} where, under nutrient fluctuations leading to drastic environmental changes, having a higher rate of phenotypic variation represents a competitive advantage, and the mean phenotype switches between the two extreme phenotypic states (\emph{i.e.} $x=0$ and $x=1$). Thus, in these regions we would expect to have higher levels of phenotypic heterogeneity, consistent with previous experimental findings~\cite{chen2018intermittent}, and to observe cells adopting bet-hedging as a survival strategy in response to drastic variation in oxygen levels. These conclusions are summarised in the table provided in Figure~\ref{summary}. 
\begin{figure}[h] \centering
\includegraphics[width=1 \linewidth]{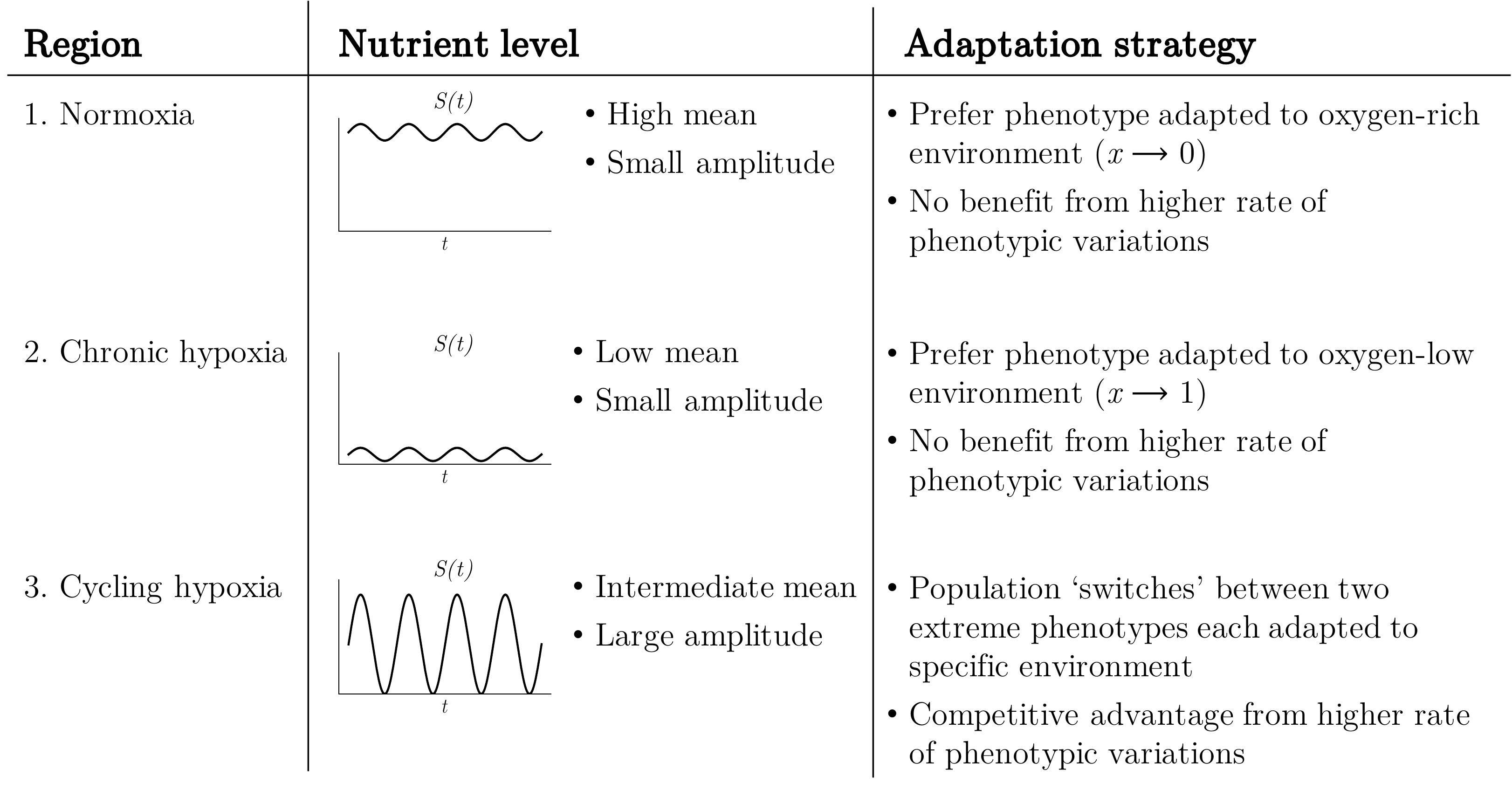}
\caption{\label{summary} Summary of the biological interpretation of our results in the context of cancer metabolism}
\end{figure}

\section{Conclusions}
\label{Conclusions}
We have presented a mathematical model for the evolutionary dynamics of two asexual phenotype-structured populations competing in periodically oscillating environments. The two populations undergo heritable, spontaneous phenotypic variations at different rates and their fitness landscape is dynamically sculpted by the occurrence of fluctuations in the concentration of a nutrient.

Our analytical results formalise the idea that when nutrient levels experience small and slow periodic oscillations, and thus environmental conditions are relatively stable, it is evolutionarily more efficient to rarely undergo spontaneous phenotypic variations. Conversely, under relatively large and fast periodic oscillations in the nutrient levels, which lead to alternating cycles of starvation and nutrient abundance, higher rates of spontaneous phenotypic variations can confer a competitive advantage, as they may allow for a quicker adaptation to changeable environmental conditions. In the latter case, our results indicate that higher levels of phenotypic heterogeneity are to be expected compared to those observed in slowly fluctuating environments. Finally, our results suggest that bet-hedging evolutionary strategies, whereby individuals switch between antithetical phenotypic states, can naturally emerge in the presence of relatively large and fast nutrient fluctuations leading to drastic environmental changes. 

We conclude with an outlook on possible extensions of the present work. The focus of this paper has been on the case where the maximum proliferation rate of the phenotypic variants best adapted to nutrient-rich environments (\emph{i.e.} the parameter $\gamma$) and the maximum proliferation rate of the phenotypic variants best adapted to nutrient-scarce environments (\emph{i.e.} the parameter $\zeta$) are the same. However, there are biological scenarios whereby the ability to survive in harsh environments comes with a fitness cost. For instance, cells are known to turn to glucose for their energy production when oxygen is in short supply, \emph{i.e.} they produce energy through anaerobic glycolysis instead of using oxidative phosphorylation that requires aerobic conditions. Since anaerobic glycolysis is far less efficient in terms of produced energy than oxidative phosphorylation~\cite{vander2009understanding}, the proliferation rate of glucose-dependent phenotypic variants might be lower than that of phenotypic variants best adapted to oxygenated environments. Therefore, it would be interesting to extend our analytical results to the case where $\zeta < \gamma$. It would also be interesting to consider cases where an integral operator is used, in place of a differential operator, to describe the effect of phenotypic variations. In this case, we expect the qualitative behaviour of the results presented here to remain unchanged when the transition between phenotypic states is modelled via Gaussian kernels of sufficiently small variance. From a mathematical point of view, this would require further development of the methods of proof presented here in order to carry out a similar asymptotic  analysis of evolutionary dynamics.

Another natural extension of the model is consideration of the feedback from the populations on the nutrient level. In fact, most existing models of evolutionary dynamics in a fluctuating environment do not account for this feedback and potentially nonlinear dynamical interactions between individuals and the surrounding environment. However, changes in the population dynamics are known to affect the outcome of interspecies competition in the presence of time variations in the availability of  nutrients \cite{okuyama2015demographic}. For instance, in the context of solid tumours, one could consider both negative feedback mechanisms that regulate population growth, such as nutrient consumption, and positive feedback mechanisms that promote the supply of nutrient, such as angiogenesis, \emph{i.e.} hypoxia-induced formation of blood vessels. Consideration of these mechanisms is expected to affect the advantages gained by each population, and therefore the phenotypic composition of the tumour, in a given environment.

An additional development of our study would be to incorporate into the model a spatial structure, as done for instance in~\cite{bouin2014travelling,lorenzi2018role,lorz2015modeling,alfaro2013travelling,mirrahimi2015asymptotic,lam2014evolution}, and let multiple nutrient sources with different inflows be distributed across the spatial domain. This would lead individuals to experience nutrient fluctuations of variable amplitudes and frequencies depending on their spatial position, thus leading to the emergence of multiple local niches whereby different phenotypic variants could be selected. Such an extension of our study would be relevant in several biological contexts. In fact, spatial niche partitioning is known to have an important impact on interspecies competition, as it promotes the coexistence between species best adapted to different local environmental conditions~\cite{hassell1994species,kneitel2004trade}. Moreover, in the context of cancer research, clinical images and histological data have revealed the existence of considerable levels of spatial heterogeneity in oxygen distribution within tumours~\cite{tomaszewski2017oxygen,michiels2016cycling}, and localised regions of cycling hypoxia and chronic hypoxia have been identified in tumour xenografts~\cite{matsuo2014magnetic}. In this regard, a spatially-stuctured version of our model could shed new light on the ways in which the interplay between spatial and temporal variability of oxygen levels may dictate the phenotypic composition and the level of phenotypic heterogeneity of malignant tumours.

\begin{acknowledgements}
AA is supported by funding from the Engineering and Physical Sciences Research Council (EPSRC) and the Medical Research Council (MRC) (grant no. EP/L016044/1) and in part by the Moffitt Cancer Center PSOC, NIH/NCI (grant no. U54CA193489). RG and ARAA are supported by Physical Sciences Oncology Network (PSON) grant from the National Cancer Institute (grant no. U54CA193489) as well as the Cancer Systems Biology Consortium grant from the National Cancer Institute (grant no. U01CA23238). ARAA and RG would also like to acknowledge support from the Moffitt Cancer Center of Excellence for Evolutionary Therapy.
\end{acknowledgements}

\bibliographystyle{spphys}       
\bibliography{Manuscript}   

\newpage
\appendix
\section{Supplementary figures}
\begin{figure}[h!] \centering
\includegraphics[width=1 \linewidth]{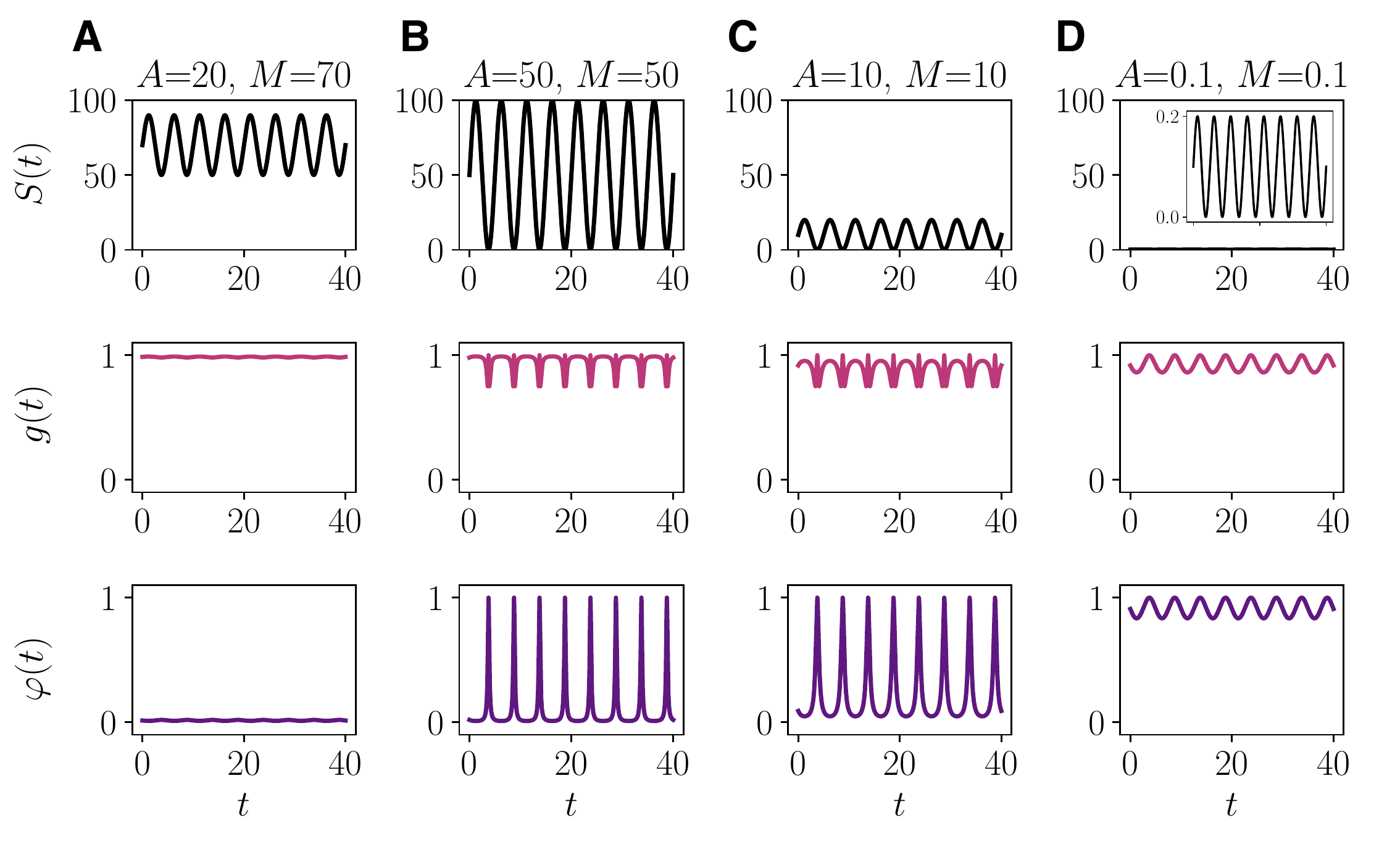}
\caption{\label{figS1} {\bf {\textsf A.}} Plots of the nutrient concentration $S(t)$ (top row) defined according to~\eqref{sinS} with $M=70$, $A=20$ and $T=5$, and the corresponding rescaled maximum fitness $g(t)$ (central row) and fittest phenotypic state $\varphi(t)$ (bottom row) defined according to~\eqref{gh}. {\bf {\textsf B} - {\textsf D.}} Same as column {\bf {\textsf A}} but with $A=M=50$ (column {\bf {\textsf B}}), $A=M=10$ (column {\bf {\textsf C}}), and $A=M=0.1$ (column {\bf {\textsf D}})}
\end{figure}

\end{document}